\documentclass[sigconf,nonacm]{acmart}

\usepackage{tikz}
\usepackage{algorithm}
\usepackage{algorithmic}
\usepackage{booktabs}
\usepackage{amsmath}
\usepackage{mathtools}
\usepackage{amsthm}
\usepackage{graphicx}
\usepackage{placeins}
\usetikzlibrary{calc,patterns.meta,intersections,shapes.misc,shapes.multipart,shapes.geometric,fpu,angles,quotes,positioning}
\theoremstyle{remark}

\newif\ifshowrevisions
\ifshowrevisions
  \definecolor{revcolor}{rgb}{0,0,1}
\else
  \definecolor{revcolor}{rgb}{0,0,0}
\fi
\newcommand{\rev}[1]{{\color{revcolor}#1}}

\setcopyright{none}
\acmDOI{}
\acmISBN{}

\begin{document}
\title{\texorpdfstring{\rev{Graph-Based Nearest-Neighbor Search}}{Graph-Based Nearest-Neighbor Search} with Recall Guarantees Using Graph Spanners}
\author{Minghao Li}
\affiliation{%
  \institution{University of Toronto}
  \city{Toronto}
  \country{Canada}
}
\email{mingh.li@mail.utoronto.ca}

\author{Raghav Mittal}
\affiliation{%
  \institution{The University of Texas at Arlington}
  \city{Arlington}
  \state{Texas}
  \country{USA}
}
\email{rxm0006@mavs.uta.edu}

\author{Sanjivni Rana}
\affiliation{%
  \institution{The University of Texas at Arlington}
  \city{Arlington}
  \state{Texas}
  \country{USA}
}
\email{sxr0277@mavs.uta.edu}

\author{Suraj Shetiya}
\affiliation{%
  \institution{IIT Bombay}
  \city{Mumbai}
  \country{India}
}
\email{surajs@cse.iitb.ac.in}

\author{Gautam Das}
\affiliation{%
  \institution{The University of Texas at Arlington}
  \city{Arlington}
  \state{Texas}
  \country{USA}
}
\email{gdas@uta.edu}

\author{Nick Koudas}
\affiliation{%
  \institution{University of Toronto}
  \city{Toronto}
  \country{Canada}
}
\email{koudas@cs.toronto.edu}

\begin{abstract}
\rev{Graph-based indexes such as} Hierarchical Navigable Small World (HNSW) graphs make nearest-neighbor retrieval efficient, but a high average recall can hide queries whose results fall below an application's quality requirement. We introduce Certify-then-Rectify (CTR), which returns an initial \rev{graph-search} result when a calibrated rule accepts it and searches for missing neighbors otherwise. The recovery search reuses the existing graph and the distances already evaluated by \rev{the initial search}. To determine how far this search should proceed, we relate metric distances to graph-path lengths through graph stretch. We develop Stretch-Bounded Expansion from Query (SBE-Q), together with Metric Bound Verification (MBV) to prune unnecessary work, and estimate its budget from paths to the fraction of true neighbors needed for the requested recall. Finite-sample tolerance bounds support this recall-aligned estimate; paths recorded during proposal recovery offer a conservative alternative to separate shortest-path computations. Independent tests certify initial acceptance and recovery outcomes, giving an end-to-end population recall guarantee even when the estimated stretch does not cover every relevant path. If recovery cannot be certified, rejected queries use exact scan. Experiments \rev{on HNSW} compare five stretch estimators, assess recovery efficiency, and evaluate compliance and throughput on databases of up to 100 million vectors\rev{. CTR is not specific to HNSW, and we also evaluate it on another graph index}.
\end{abstract}

\maketitle
\raggedbottom

\section{Introduction}\label{sec:introduction}
Nearest-neighbor search supports vector retrieval by finding database records close to a query. Exact search becomes expensive as the database grows because it evaluates every record. Approximate nearest-neighbor (ANN) indexes reduce this work by examining only part of the database. \rev{Graph-based indexes~\cite{nsg,diskann}, such as} Hierarchical Navigable Small World (HNSW) graphs~\cite{malkov2018efficient,acorn}\rev{,} are widely used for their search performance, but their recall depends on the graph, the search budget, and the query. As a result, high average recall can coexist with a substantial fraction of queries that miss an application's requested recall target.

\definecolor{ctrteal}{HTML}{3A8C8C}
\definecolor{ctrnavy}{HTML}{34558B}
\definecolor{ctrorange}{HTML}{C9772E}
\definecolor{ctrwine}{HTML}{A8444A}
\definecolor{ctrgreen}{HTML}{5E9A4B}
\begin{figure}[t]
\centering
\begin{tikzpicture}[
  >=latex, font=\scriptsize,
  ctr/.style={rectangle split, rectangle split parts=2, rounded corners=2.5pt,
    draw=#1, line width=0.6pt,
    rectangle split part fill={#1!6,#1!88},
    align=center, text width=2.0cm, inner sep=1.7pt},
  flow/.style={->,draw=black!55,line width=0.6pt},
  edge label/.style={font=\scriptsize,inner sep=1pt,text=black!70},
  inset/.style={rounded corners=2.5pt,draw=black!15,fill=white,line width=0.5pt}
]
\node[ctr=ctrteal] (search) at (1.5,0)
  {\textcolor{ctrteal!85!black}{Index (e.g., HNSW)}\nodepart{two}\textcolor{white}{\bfseries Graph search}};
\node[ctr=ctrnavy] (accept) at (4.1,0)
  {\textcolor{ctrnavy}{Fit score + LTT test}\nodepart{two}\textcolor{white}{\bfseries Acceptance $C(q)$}};
\node[ctr=ctrgreen,text width=1.75cm] (result) at (7.3,0)
  {\textcolor{ctrgreen!85!black}{Certified}\nodepart{two}\textcolor{white}{\bfseries Return result}};
\node[ctr=ctrorange] (recovery) at (4.1,-1.12)
  {\textcolor{ctrorange!85!black}{Estimate $t'$}\nodepart{two}\textcolor{white}{\bfseries Recovery}\\[-1pt]\textcolor{white}{\tiny SBE-Q + MBV}};
\node[ctr=ctrwine,text width=1.75cm] (certify) at (7.3,-1.12)
  {\textcolor{ctrwine}{LTT recovery test}\nodepart{two}\textcolor{white}{\bfseries Certification}\\[-1pt]\textcolor{white}{\tiny fail: exact scan}};
\draw[flow] (0.0,0) -- node[edge label,above] {$q$} (search.west);
\draw[flow] (search) -- (accept);
\draw[flow] (accept) -- node[edge label,above,text=ctrgreen!85!black] {accept} (result);
\draw[flow] (accept) -- node[edge label,right] {reject} (recovery);
\draw[flow] (recovery) -- (certify);
\draw[flow] (certify) -- node[edge label,right] {pass} (result);
\begin{scope}[shift={(0.04,-1.86)}]
  \draw[inset] (0,0) rectangle (2.52,1.40);
  \node[inner sep=0pt] at (1.26,0.70)
    {\includegraphics[height=1.32cm,trim=497 76 0 44,clip]{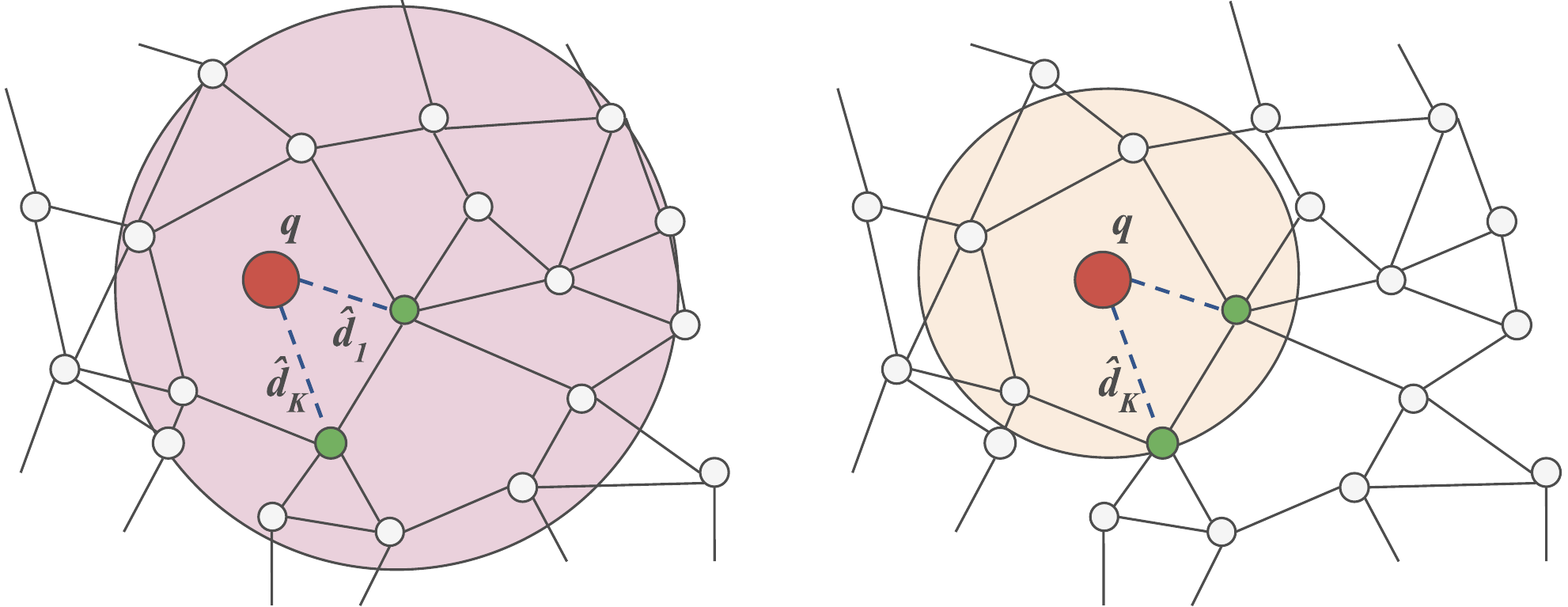}};
\end{scope}
\draw[black!30,dotted,line width=0.6pt] (2.56,-1.12) -- (recovery.west);
\begin{scope}[shift={(0.12,-2.82)}]
  \draw[inset] (-0.08,-0.02) rectangle (8.22,0.86);
  \node[font=\tiny\bfseries,anchor=west,text=ctrwine] at (0.0,0.73) {LTT test};
  \foreach \columnindex in {0,...,15} {
    \foreach \rowindex/\rowy in {0/0.40,1/0.23} {
      \pgfmathtruncatemacro{\cellid}{\rowindex*16+\columnindex}
      \ifnum\cellid=4 \def\cellcolor{ctrwine!70}\else\ifnum\cellid=25 \def\cellcolor{ctrwine!70}\else\def\cellcolor{ctrgreen!35}\fi\fi
      \fill[\cellcolor] ({0.04+\columnindex*0.17},\rowy) rectangle ++(0.13,0.13);
    }
  }
  \node[font=\tiny,anchor=west,text=black!75] at (0.0,0.09) {$m$ calibration queries; $x$ with recall $<\tau$ (red)};
  \draw[flow] (2.86,0.42) -- (3.36,0.42);
  \node[font=\tiny,align=center,text=black!80] at (4.47,0.42)
    {binomial p-value\\[1pt]$p(x,m;\epsilon)=\Pr[\mathrm{Bin}(m,\epsilon)\le x]$};
  \draw[flow] (5.58,0.42) -- (6.04,0.42);
  \draw[black!40,line width=0.4pt] (6.14,0.18) -- (7.68,0.18);
  \foreach \barx/\barheight/\barcolor/\barlabel in
      {6.2/0.08/ctrgreen!55/1, 6.56/0.20/ctrgreen!55/2, 6.92/0.48/ctrwine!55/3, 7.28/0.40/ctrwine!55/4} {
    \fill[\barcolor] (\barx,0.18) rectangle ++(0.24,\barheight);
    \node[font=\tiny,inner sep=0pt,text=black!75] at ({\barx+0.12},0.08) {$\lambda_{\barlabel}$};
  }
  \draw[ctrnavy,dashed,line width=0.5pt] (6.14,0.52) -- (7.68,0.52);
  \node[font=\tiny,anchor=west,text=ctrnavy,inner sep=1pt] at (7.68,0.52) {level};
  \node[font=\tiny,anchor=west,text=ctrgreen!85!black,inner sep=1pt] at (7.68,0.28) {pass};
\end{scope}
\draw[black!30,dotted,line width=0.6pt] (7.3,-1.96) -- (certify.south);
\end{tikzpicture}
\caption{\rev{CTR architecture. Upper bands are prepared offline; lower bands run at query time. Rejected queries use graph recovery (SBE-Q with MBV) if a second LTT test certifies it, and exact scan otherwise. Insets: SBE-Q's search region (left) and the LTT test (bottom).}}
\Description{A query enters a search on a graph index such as HNSW. The acceptance rule, fitted with a score and an LTT test offline, either returns the initial result or rejects the query. Rejected queries go to graph recovery with an estimated stretch, followed by a second LTT test that certifies recovery; if certification fails, exact scan is used. Both branches return a result. An inset, cropped from Figure 2, shows the shaded SBE-Q search region around the query and its initial candidates. A second inset shows the LTT test: calibration queries marked as meeting or missing the recall target, a binomial p-value, and p-values of four candidates compared with the test level; two pass.}
\label{fig:framework_small}
\end{figure}

Exact retrieval is useful when a missed neighbor changes the subsequent computation. For example, mining difficult training examples can require retrieving nearby examples that a model handles incorrectly~\cite{intelligence2025pi06vlalearnsexperience,chen2023focalformer3d}. Missing these neighbors can reduce the coverage of the training sample. Similarly, when retrieved entities provide premises for a symbolic or language-model computation, a missing entity can leave that computation with incomplete input~\cite{wiatrak2023proxy,lin2023knn}. Exact nearest-neighbor retrieval addresses these omissions within the chosen representation and metric; it does not by itself establish the quality or safety of the downstream model.

Let $\tau\in(0,1]$ denote the requested recall target and $\epsilon$ the tolerated fraction of queries below that target. We require recall at least $\tau$ on at least a $1-\epsilon$ fraction of the workload. Increasing \rev{the index's} search budget can improve empirical recall, and learned stopping rules can adapt effort to query difficulty~\cite{darth,learnedtermination1,adaptivebeam}. The remaining challenge is to establish that the returned results meet the population requirement, while specifying what happens when the initial search provides insufficient evidence to return its answer. Sending every such query to exact scan is a reliable option, but it discards the graph and the search work already available.

We propose \textbf{Certify-then-Rectify (CTR)} to combine selective acceptance with recovery on the existing \rev{graph index}. \rev{The index} first returns candidate neighbors and a trace of evaluated distances. A calibrated acceptance rule decides whether to return this initial result or perform additional search. For queries that need more work, CTR reuses the candidates and evaluated distances to search for missing neighbors. We call this additional search \emph{recovery}. An independent calibration stage tests the result of recovery; if it cannot be certified, CTR uses exact scan for rejected queries.

\textbf{The statistical objective matters:} certifying the fraction of queries that attain a recall target with prescribed confidence differs from controlling expected recall loss or evaluating prediction-based stopping empirically. Section~\ref{sec:exp-e2e} distinguishes the statistical guarantees of the compared methods.

Recovery must decide how far to search for neighbors that \rev{the initial search} may have missed. When the initial result contains $k$ distinct candidates, every true top-$k$ neighbor is no farther from the query than the most distant candidate. This gives a distance bound on where missing neighbors can lie. However, reaching a nearby record through graph edges may require following a long path. \emph{Graph stretch} measures this detour by comparing the length of a shortest graph path with the metric distance between its endpoints. We use a stretch bound to turn the distance bound into a limit on the path lengths explored during recovery.

We first derive this limit for a search starting at the nearest initial candidate. We then develop \textbf{Stretch-Bounded Expansion from Query (SBE-Q)}, which temporarily connects the query to multiple candidates and starts recovery from the query itself. Within this search, \textbf{Metric Bound Verification (MBV)} uses already evaluated distances and the triangle inequality to avoid unnecessary distance computations and path extensions.

To set this recovery limit, we need a usable stretch value. \rev{Common graph constructions, including HNSW, provide} no general stretch bound, and the largest detours between arbitrary records may have little relevance to the queried neighbors. We therefore sample queries sent to recovery and obtain their exact nearest neighbors offline. For each query, we compute a stretch sufficient to reach enough of these neighbors to meet the recall target $\tau$. An order-statistic tolerance bound then selects a shared operating stretch that is sufficient for a specified fraction of future rejected queries, with a stated confidence level.

We also estimate stretch from paths recorded during a fixed recovery procedure. These paths upper-bound shortest-path lengths, providing a conservative alternative that avoids separate shortest-path computations. We compare both query-based estimates with estimates from sampled database pairs, including extreme-value fitting. The selected stretch determines how far graph recovery searches. Independent calibration then tests the result of recovery; together with the acceptance test, it establishes the population recall guarantee. Exact scan handles rejected queries when graph recovery cannot be certified.

We make three contributions:
\begin{itemize}
\item \textbf{Efficient recovery on \rev{graph indexes}.} We derive SBE-Q from candidate-centered expansion and combine it with MBV. We establish conditions for recovering the requested fraction of true neighbors and prove that recovery retains correct initial neighbors.
\item \textbf{Recall-aligned stretch estimation.} We derive a query-path score sufficient for a requested recall, apply finite-sample tolerance bounds to construct recovery budgets, and develop a conservative recorded-path alternative.
\item \textbf{End-to-end population certification.} Independent acceptance and recovery tests bound the final query-failure probability, with an exact-scan default. Experiments compare five estimators, assess recovery mechanisms, and evaluate compliance and throughput up to 100M vectors\rev{, mainly on HNSW and also on another graph index}.
\end{itemize}

\rev{We instantiate our analysis on HNSW, the most widely adopted graph index, which shows state-of-the-art performance across a range of tasks~\cite{malkov2018efficient,ann1}; unless stated otherwise, the index is HNSW. CTR is not specific to HNSW: it requires only a metric-weighted base graph and a search procedure that exposes the query distances it evaluates, so it applies to other graph indexes. Our experiments also test CTR on another graph index (see Appendix~\ref{app:nsg}). Experimental evaluation on filtered search is left as future work.}

\section{Preliminaries}\label{sec:preliminaries}
\rev{We instantiate the analysis on HNSW, but CTR applies to other graph indexes; Section~\ref{sec:graph-index-interface} states what it requires from an index.} We first define recall and the metric and tie-breaking conventions used throughout the paper. We then describe the HNSW index, its search procedure, and the evaluated distances retained for recovery. Finally, we define graph stretch, which relates metric distances to shortest-path lengths in the graph.

\subsection{Retrieval Objective and Metric Conventions}
Let $\mathcal X=\{x_1,\ldots,x_n\}$ be a database of records with distinct IDs and associated vectors in a $d$-dimensional metric space. We write $dist(x,y)$ for the distance between their vectors and identify records with their IDs in set operations. For a query $q$ and $1\le k\le n$, let $\mathcal N_k^*(q)\subseteq\mathcal X$ be the canonical top-$k$ set obtained by ranking records by distance to $q$, with ties broken by ID.

For a result $R\subseteq\mathcal X$ with $|R|\le k$ and recall target $\tau\in(0,1]$, define
\begin{align}
\operatorname{recall}(R,q)&=\frac{|R\cap\mathcal N_k^*(q)|}{k},\label{eq:recall}\\
L_\tau(R,q)&=\mathbf1\{\operatorname{recall}(R,q)<\tau\}.\label{eq:recall-loss}
\end{align}
A result is \emph{compliant} if $L_\tau(R,q)=0$. For a fixed retrieval procedure and a fresh workload query $Q$, the expected loss equals the population probability of noncompliance. We seek a bound on this risk with a specified confidence level over the independent calibration sample.

During search, a result heap retains the best $k$ distinct candidates seen so far, using the same distance-and-ID order as the reference set. Distinct records may have identical vectors and therefore zero distance between them. They remain distinct possible neighbors under the ID convention, and the triangle inequalities used by recovery still hold.

Recovery uses triangle inequalities to bound query distances that have not yet been evaluated. All metric bounds therefore use unsquared distances; squared Euclidean distance does not satisfy the triangle inequality. For inner-product retrieval, we use a ranking-preserving asymmetric L2 transformation to obtain Euclidean geometry (Section~\ref{sec:t2i-transform}). Appendix~\ref{app:implementation} gives further details on duplicates, ties, and numerical arithmetic.

\subsection{\rev{Graph Index Interface,} Search, and Trace}\label{sec:graph-index-interface}
\rev{CTR requires two components from the index: a base graph $G_0$ whose vertex set is the database $\mathcal X$, and a search procedure that explores $G_0$ and returns candidates together with the query distances it evaluated. We describe this interface through HNSW but other graph indexes fit it too.}

An HNSW index consists of nested graphs $\mathcal H=(G_0,\ldots,G_{l_{\max}})$, where $l_{\max}$ is the highest layer, $G_i=(V_i,E_i)$, $V_0=\mathcal X$, and $V_i\subseteq V_{i-1}$. Each record receives a random maximum layer from a distribution with an exponentially decaying tail and appears in all layers below it. Neighbor selection balances proximity and spatial diversity under degree limits controlled by $M$; construction search effort is controlled by $ef_c$~\cite{malkov2018efficient}.

For a query $q$, HNSW starts at an entry point in the highest layer, moves greedily toward closer records, and descends to the bottom layer. It then runs a bounded candidate search with beam width $ef_s$ to produce the initial result $H(q)$. During this search, the candidate frontier holds discovered records awaiting expansion; expanding a record examines its neighbors. The returned set may omit true neighbors or contain fewer than $k$ records. CTR retains both the result and the search trace
\[
\mathcal T(q)=\{x\in\mathcal X:dist(q,x)\text{ was evaluated}\}.
\]
The trace includes every record whose query distance was evaluated, including neighbors that were never expanded. CTR caches these distances so recovery can continue the search without evaluating them again.

The candidates also provide a useful bound on where missing neighbors can lie. We call the distance to the current $k$-th best candidate the \emph{result cutoff}: a farther record cannot enter the current top-$k$ result. We denote the initial cutoff by $\hat d_k$ and its evolving value during recovery by $D$. Both are $+\infty$ until $k$ distinct candidates are available. If $d_k^*$ is the distance to the true $k$-th neighbor, then $D\ge d_k^*$, including at the initial value $\hat d_k$. Thus every true neighbor lies within metric distance $D$ of the query. The challenge is to reach those records through graph edges.

\subsection{Graph Stretch}
A metric radius alone does not specify how far a graph search must proceed. Two nearby records may be connected only by a long detour, so recovery must account for path length as well as query distance. We assign each edge $(u,v)$ in the bottom-layer graph $G_0$ its nonnegative metric length $dist(u,v)$ and write $d_G(u,v)$ for the shortest weighted path length from $u$ to $v$. An unreachable pair has path length $+\infty$. A weighted graph ball of radius $r$ around $u$ contains the vertices reachable by paths of length at most $r$; Dijkstra's algorithm explores such a ball. Graph stretch bounds the detour between two records:

\begin{definition}[Graph spanner]
A metric-weighted graph $G=(V,E)$ is a $t$-spanner if a finite $t\ge1$ satisfies, for every ordered pair $u,v\in V$,
\[
dist(u,v)\le d_G(u,v)\le t\,dist(u,v).
\]
The value $t$ is its stretch factor.
\end{definition}
For distinct metric locations, the empirical all-pairs stretch is
\[
t_{\mathrm{emp}}=\max_{u\ne v\in\mathcal X}\frac{d_{G_0}(u,v)}{dist(u,v)}.
\]
For coincident records, we assign the ratio value one if a zero-length graph path exists and $+\infty$ otherwise. A finite all-pairs stretch therefore requires reachability of every ordered pair. Estimation and recovery use the same directed edges, or the same explicitly symmetrized graph; sampling only reachable pairs would describe a different population.

If a stretch bound were known, it would translate the metric region containing true neighbors into a graph region that recovery could explore. \rev{Common graph constructions, including HNSW, supply} no known constant stretch bound on arbitrary inputs. CTR uses an estimated \emph{operating stretch} for this purpose and independently tests whether the resulting recovery procedure meets the population recall requirement.

\section{CTR Architecture}\label{sec:architecture}
CTR decides which initial \rev{index} results to return immediately and which queries require additional search (Figure~\ref{fig:framework_small}). It prepares two components offline: one acceptance rule and one recovery policy. These components remain fixed during query processing. We first describe how they process a query, then explain how offline data are used to prepare them.

\paragraph{Query execution.}\label{sec:decision-rule}
For a new query $q$, \rev{the index search} produces an initial result $H(q)$ and a search trace $\mathcal T(q)$. The stored acceptance rule uses information from this result and trace to make a binary decision. We write $C(q)=1$ when it accepts the query: CTR returns $H(q)$ immediately. When $C(q)=0$, it rejects the query: CTR sends it to recovery before returning an answer. Thus an \emph{accepted query} is one answered by the initial \rev{search} result, and a \emph{rejected query} is one sent to additional search. Rejection does not establish that the initial result is wrong; the rule may send an already compliant result to recovery.

The recovery policy specifies either SBE-Q with MBV, the graph recovery algorithm developed in Section~\ref{sec:recovery}, or exact scan. Graph recovery reuses the initial candidates, evaluated distances, and \rev{base graph $G_0$} to search for missing neighbors. Exact scan examines all database records and returns the true top-$k$ set. The choice is made offline and remains fixed during query processing.

\paragraph{Offline preparation.}
We prepare both components using query samples with exact nearest-neighbor answers. First, we train a score $s(q)$ from $H(q)$ and $\mathcal T(q)$ to predict whether the initial result reaches recall $\tau$. We then apply Learn-then-Test (LTT) acceptance calibration on a separate query sample. This calibration rule selects $C(q)$ from candidate thresholds on $s(q)$ and bounds the fraction of accepted results below $\tau$ (Section~\ref{sec:ltt-acceptance}). If no candidate threshold passes, we set $C(q)=0$ for every query.

With $C$ fixed, we estimate an operating stretch $t'$ from a separate sample of rejected queries, as described in Section~\ref{sec:sampling-stretch}. This value sets the path-length budget for SBE-Q with MBV. Together with a fixed attachment rule and traversal settings, it defines a candidate graph recovery procedure $\mathcal R_\lambda$, where $\lambda$ indexes the configuration (e.g., the operating stretch $t'$, maximum number of query-to-candidate connections $M_0$, and deterministic limits on beam width, distance evaluations, or vertex expansions). These choices are fixed before recovery calibration.

We assess $\mathcal R_\lambda$ using the \emph{LTT recovery-certification test} in Section~\ref{sec:recovery-calibration}. On independent calibration queries with $C(q)=0$, we run $\mathcal R_\lambda$ from the retained $H(q)$ and $\mathcal T(q)$ and compare its results with the exact answers. The test uses the number of results with recall below $\tau$ to determine whether the data support a population failure bound $\epsilon_r$, the tolerated fraction of rejected queries whose recovery results fall below $\tau$. It uses a binomial p-value and accounts for selection when several recovery configurations are tested.

If at least one configuration passes, we select a certified $\mathcal R_{\hat\lambda}$ and save its operating stretch, attachment rule, and traversal settings for deployment. Every future query with $C(q)=0$ then runs SBE-Q with MBV using that configuration. If none passes, every rejected query uses exact scan. Section~\ref{sec:end-to-end} combines the acceptance and recovery calibration guarantees into a bound on the fraction of future workload queries whose final results meet $\tau$, with a prescribed confidence over calibration. Section~\ref{sec:experiments} specifies the data allocation.

\section{Stretch-Bounded Recovery}\label{sec:recovery}
When CTR rejects an initial result, recovery must find additional neighbors without repeating the entire search. The current result cutoff $D$ is the distance from the query to the $k$-th closest distinct candidate evaluated so far; we set $D=+\infty$ if fewer than $k$ candidates are available. This cutoff may decrease as recovery finds closer candidates. It bounds the query distance to every true top-$k$ neighbor, but recovery reaches those neighbors through graph paths. We first show how a supplied stretch converts this bound into an expansion radius around an existing candidate. Moving the source to the query then yields SBE-Q, and metric bounds make its traversal more efficient. We treat stretch as a supplied parameter here; Section~\ref{sec:sampling-stretch} explains how to choose it from data.

\subsection{Expansion from a Nearest Candidate}\label{sec:sbe-nn-main}
An initial candidate is already a vertex of the database graph, making it a natural source for additional traversal. We choose the nearest returned candidate $a$ and write $\hat d_1=dist(q,a)$ for its query distance. The result heap retains the best $k$ distinct evaluated candidates, with cutoff $D$. If paths from $a$ to each true neighbor $x^*$ have stretch at most $t$, the triangle inequality gives
\begin{equation}\label{eq:sbe-nn-main}
\begin{aligned}
d_{G_0}(a,x^*)&\le t\,dist(a,x^*)\\
&\le t\bigl(dist(q,a)+d_k^*\bigr)\\
&\le t\bigl(dist(q,a)+D\bigr).
\end{aligned}
\end{equation}
We call the resulting procedure \emph{Stretch-Bounded Expansion from a Nearest Candidate} (SBE-NN). It explores the graph from $a$ up to weighted path length $t(dist(q,a)+D)$. For a common supplied $t$ and cutoff $D$, the nearest candidate deterministically gives the smallest such radius among the initial candidates. SBE-NN requires no new attachments. Appendix~\ref{sec:sbe-nn} states the formal SBE-NN containment result (Theorem~\ref{thm:sbe-nn}), with its proof in Appendix~\ref{app:proof-sbe-nn}.

\subsection{Expansion from the Query}\label{sbeq}
The SBE-NN radius includes the distance from $q$ to its source $a$ as well as the distance to a missing neighbor. This suggests centering recovery at the query itself. Unlike $a$, however, $q$ is not initially connected to the index. SBE-Q introduces it as a temporary source and attaches it to multiple evaluated candidates, providing several routes into the database graph.

\label{sec:graph-interface}
We choose an attachment pool $\mathcal H_{found}\subseteq\mathcal T(q)$ and a connection limit $M_0\ge1$. We then add outgoing metric-weighted edges from $q$ to the nearest records in this pool, using at most $M_0$ connections. The augmented graph $G'_q$ leaves the stored index unchanged. Adding these connections cannot increase shortest-path distances from $q$. The attachment rule is part of the recovery configuration and is used consistently in estimation, certification, and deployment.

SBE-Q uses the supplied stretch $t'$ to explore the augmented graph $G'_q$ up to weighted path length $t'D$. As better candidates reduce $D$, the traversal tightens this budget. If $t'$ bounds the graph-path stretch from $q$ to each true neighbor $x^*$, then
\begin{equation}\label{eq:sbe-q-bound}
d_{G'_q}(q,x^*)\le t' dist(q,x^*)\le t'd_k^*\le t'D.
\end{equation}
Appendix~\ref{app:query-source} states the formal SBE-Q containment result (Theorem~\ref{thm:sbe-q}), with its proof in Appendix~\ref{app:proof-sbe-q}. For a target $\tau<1$, recovery needs only enough of them to meet the requested recall; Section~\ref{sec:recovery-properties} gives the corresponding weaker path condition. Because the query has its own attachments, a stretch measured between database records need not bound paths from this new source. This observation motivates the query-path estimator in Section~\ref{sec:query-proposals}.

With equal operating values $t'=t$, SBE-Q omits the additional $dist(q,a)$ term from its radius and gains multiple entry routes. These changes jointly define the recovery design evaluated in Section~\ref{sec:exp-ablation}. The numerical radii alone do not order the explored graph regions or the numbers of visited vertices.

\subsection{Metric Bound Verification}\label{sec:mbv}
Even within the chosen graph ball, many visited records cannot improve the result. Evaluating their query distances and extending their paths can dominate recovery cost. MBV uses the triangle inequality to avoid both kinds of work: a metric lower bound can rule out a record as a candidate, and the remaining path budget can rule out further traversal through it. Figure~\ref{fig:recovery-geometry} illustrates these complementary tests.

\begin{figure}[t]
\centering
\includegraphics[width=0.78\columnwidth]{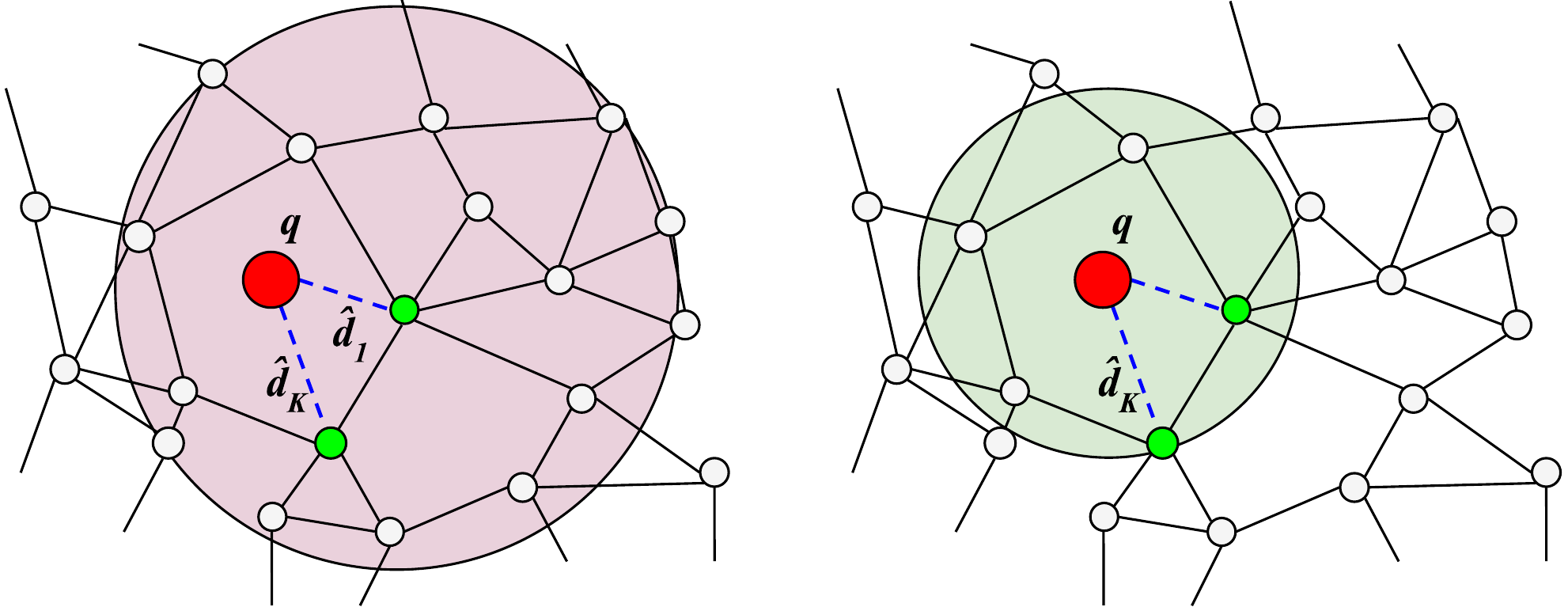}\\[3pt]
\includegraphics[width=0.78\columnwidth]{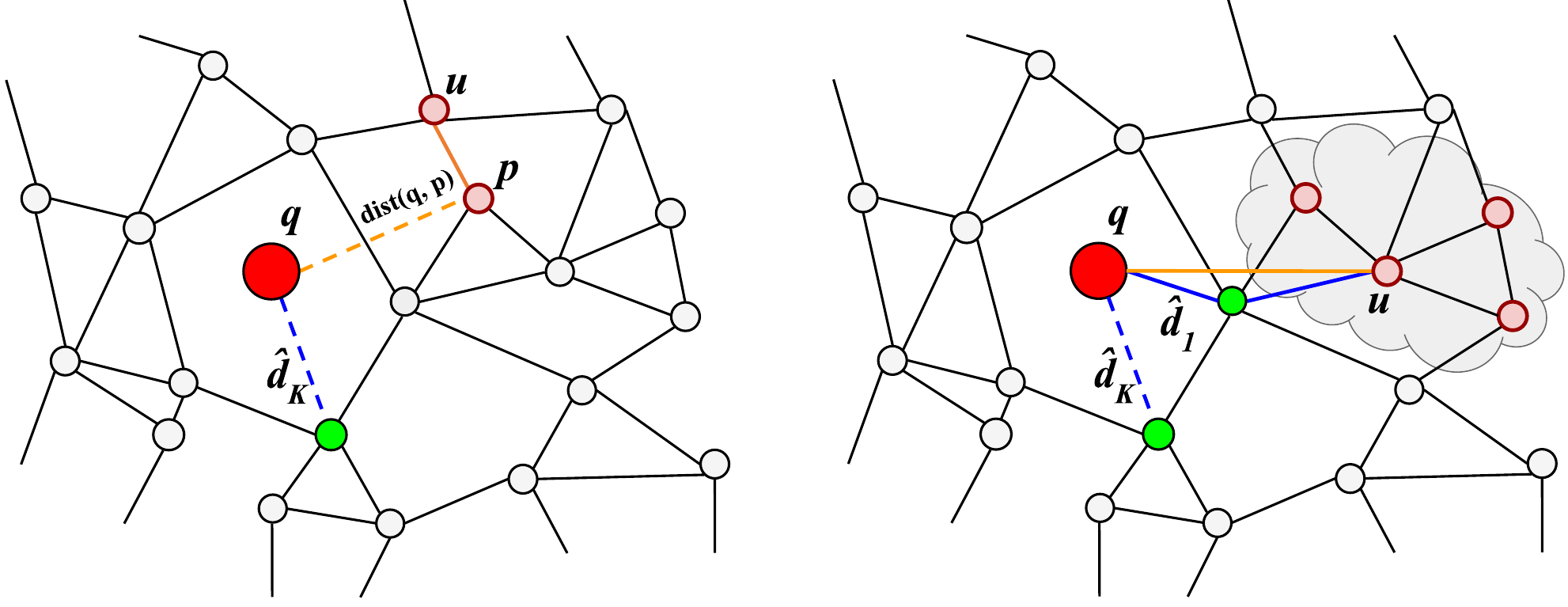}
\caption{Recovery geometry. Top: single-candidate SBE-NN and multiply attached SBE-Q. Bottom: metric lower-bound and path-budget pruning. Shading illustrates graph-search budgets schematically.}
\Description{The first pair of diagrams compares candidate-centered and query-centered expansion. The second pair illustrates metric lower-bound and path-budget pruning.}
\label{fig:recovery-geometry}\label{fig:sbe}
\end{figure}

\paragraph{Lower-bound pruning.}\label{subsec:recursive_pruning}
For each vertex $u$, MBV maintains a lower bound $LB(u)\le dist(q,u)$, initialized to zero. When it examines an edge $(p,u)$ of metric weight $w$, it propagates a bound from $p$. If the exact distance $dist(q,p)$ is cached, the reverse triangle inequality gives a tighter bound:
\begin{equation}\label{eq:lower-bound}
b(p,u)=
\begin{cases}
|dist(q,p)-w|,&p\text{ cached},\\
\max\{0,LB(p)-w\},&\text{otherwise}.
\end{cases}
\end{equation}
MBV updates $LB(u)$ to $\max\{LB(u),b(p,u)\}$, combining bounds from different predecessors by their maximum. If $LB(u)>D$, the record cannot improve the current result, so MBV skips its exact query-distance evaluation. The traversal still considers paths through its outgoing edges unless the separate path-budget test applies: a record outside the current metric radius may lie on a path to a true neighbor. Any exact distance already cached in $\mathcal T(q)$ is reused.

\paragraph{Path-budget pruning.}\label{subsec:elliptical_pruning}
Let $g(u)$ denote the length of the queued path from $q$ to $u$ that the traversal is currently processing. A continuation within the total budget $t'D$ has at most $t'D-g(u)$ length remaining. If it reaches an answer whose query distance is at most $D$, the triangle inequality requires $LB(u)\le D+t'D-g(u)$. MBV can therefore skip outgoing path extensions when
\begin{equation}\label{eq:ellipse-pruning}
g(u)+LB(u)>(t'+1)D.
\end{equation}
The test applies to extensions of the processed path; other paths can still reach vertices that have not been settled by the traversal. Both pruning rules use strict inequalities so that tied boundary candidates remain eligible. When $D=+\infty$, the finite-budget tests leave all paths eligible.

\subsection{Recovery Properties}\label{sec:recovery-properties}
The preceding bounds cover every true neighbor, but an approximate recall target requires only a fraction of them. This distinction lets us characterize the budget needed by an individual query. In $G'_q$, we sort the shortest-path lengths to its true top-$k$ neighbors as $\ell_{(1)}\le\cdots\le\ell_{(k)}$. For $h=\lceil\tau k\rceil$, we define
\begin{equation}\label{eq:query-stretch-score}
Z_\tau(q)=\max\{1,\ell_{(h)}/d_k^*(q)\}.
\end{equation}
For example, $k=100$ and $\tau=0.9$ require paths to 90 true neighbors. Since $D\ge d_k^*$ throughout recovery, $t'\ge Z_\tau(q)$ places these paths within the budget $t'D$. The following proposition establishes that MBV preserves them. When $d_k^*=0$, the score is one if $\ell_{(h)}=0$ and infinite otherwise.
\begin{proposition}[Partial-recall containment]\label{prop:partial-recall}
Fix the finite augmented graph, attachment rule, finite stretch $t'\ge1$, and exact metric weights. With conservative lower bounds and no additional search-work limit, SBE-Q with MBV returns recall at least $\tau$ whenever $t'\ge Z_\tau(q)$.
\end{proposition}
The proof follows paths to at least $\lceil\tau k\rceil$ true neighbors through both pruning tests (Appendix~\ref{app:proof-sampling}). This is a sufficient condition: recovery may succeed with a smaller stretch if the initial result already contains some true neighbors or its cutoff permits a larger search region. Even when the budget is insufficient, recovery retains the true neighbors already found:

\begin{proposition}[Recall cannot decrease during recovery]\label{prop:nondegradation}
Suppose recovery seeds its distinct-ID heap with all initial candidates (retaining the best $k$), uses the same exact distance-and-ID order, and subsequently only offers additional valid database candidates. Its final recall is at least its initial recall, even if the supplied stretch bound is too small. The statement also holds for filtered search, where only records satisfying a specified predicate are eligible answers, provided that both results contain only eligible IDs.
\end{proposition}
\begin{proof}
Every true top-$k$ ID precedes every non-target ID in the fixed total order. A true target already in the heap cannot be evicted: at most $k-1$ database IDs precede it. Thus every initially retrieved true target is retained. For fewer than $k$ filtered matches, no valid matching ID is evicted at all. This proves non-decreasing recall but does not prove that a missing target is recovered.
\end{proof}
Algorithm~\ref{alg:sbe_mbv_corrected} gives the complete traversal, distance caching, and bound propagation. Candidates may add deterministic limits on frontier beam width, distance evaluations, or vertex expansions. The limits, counting conventions, frontier-truncation rules, and termination behavior are fixed before calibration and unchanged during deployment. At termination, the candidate returns its current best-$k$ result heap or invokes a specified exact fallback. Frontier limits do not truncate the result heap, so returning it preserves Proposition~\ref{prop:nondegradation}. The partial-recall and exact-recovery guarantees require no such additional limits; Section~\ref{sec:recovery-calibration} certifies the complete limited procedure.

In the algorithm, $\textsc{Offer}(R,u,d)$ inserts or updates a database ID in the result heap $R$ and retains the best $k$ distinct IDs. The function $D(R)$ returns its cutoff, or $+\infty$ if fewer than $k$ records are present. We write $G'$ for the augmented graph $G'_q$.
\begin{footnotesize}
\begin{algorithm}[t]
\caption{Stretch-Bounded Expansion with MBV}\label{alg:sbe_mbv_corrected}\label{alg:recovery-summary}
\begin{algorithmic}[1]
\REQUIRE Virtual source $q$, finite metric-weighted graph $G'$, fixed $1\le t'<+\infty$, $1\le k\le n$, candidates with exact cached distances
\ENSURE Recall $\ge\tau$ under Proposition~\ref{prop:partial-recall}; exact top-$k$ under Lemma~\ref{lem:pruning}
\STATE $R\gets$ distinct best $k$ initial database candidates; $D\gets D(R)$
\STATE $g[q]\gets0$; $g[v]\gets+\infty$ for all other vertices
\STATE $A\gets\emptyset$ (settled vertices); $Q\gets\{(0,q)\}$ (min-heap)
\STATE $LB[q]\gets0$; $LB[v]\gets0$ otherwise; cache $dist(q,q)=0$
\WHILE{$Q$ is not empty}
  \STATE Pop $(h,u)$ of smallest path length from $Q$
  \IF{$u\in A$ or $h\ne g[u]$}
    \STATE \textbf{continue} \COMMENT{Settled or stale queue entry}
  \ENDIF
  \IF{$h>t'D$}
    \STATE \textbf{break}
  \ENDIF
  \STATE $A\gets A\cup\{u\}$
  \IF{$u\ne q$ and ($u$ is cached or $LB[u]\le D$)}
    \STATE Reuse or compute $d=dist(q,u)$; cache $d$; $LB[u]\gets d$
    \STATE $\textsc{Offer}(R,u,d)$; $D\gets D(R)$
  \ENDIF
  \IF{$D<+\infty$ and $h+LB[u]>(t'+1)D$}
    \STATE \textbf{continue} \COMMENT{No qualifying continuation within budget}
  \ENDIF
  \FORALL{outgoing edges $(u,v)$ of weight $w$ with $v\notin A$}
    \STATE $b\gets\max\{0,LB[u]-w\}$
    \IF{$dist(q,u)$ is cached}
      \STATE $b\gets|dist(q,u)-w|$
    \ENDIF
    \STATE $LB[v]\gets\max\{LB[v],b\}$
    \IF{$h+w<g[v]$}
      \STATE $g[v]\gets h+w$; push $(g[v],v)$ into $Q$
    \ENDIF
  \ENDFOR
\ENDWHILE
\RETURN $R$
\end{algorithmic}
\end{algorithm}
\end{footnotesize}

A vertex is settled when the traversal processes its current best queued path; an older queue entry is stale if its length differs from the current label. The algorithm tracks settled vertices separately from cached distances, and the virtual query never enters the result.\footnote{The recovery method also applies to filtered search. Its result heap contains only IDs satisfying the filter, while traversal may use all graph vertices.}

Together, these properties identify what stretch estimation must accomplish. A larger stretch gives recovery more opportunity to reach missing neighbors, while an underestimate may leave some undiscovered without discarding correct initial results. The score $Z_\tau$ expresses a sufficient budget for one query; the next section uses scores from sampled queries to choose a budget for the workload. The independent recovery test will then assess the outputs of the complete procedure.

\section{Stretch Estimation for Recovery}\label{sec:sampling-stretch}
SBE-Q needs an operating stretch before it can run, yet the path condition in Proposition~\ref{prop:partial-recall} involves true neighbors that are unknown at deployment. We therefore estimate a shared stretch offline from sampled paths. The choice of observations matters: stretches between arbitrary database records describe general graph connectivity, whereas $Z_\tau$ describes the paths recovery needs for its recall target. We first introduce pair-based estimates and their coverage interpretation, then apply the same sampling principle to query scores. Table~\ref{tab:estimators} summarizes the five alternatives. All estimates are fixed before recovery certification.

\subsection{Graph-Pair Observations and Coverage}
Computing all-pairs stretch would require examining far more paths than is practical for a large index. A direct alternative is to sample ordered pairs of distinct database IDs independently and uniformly with replacement from the completed bottom-layer graph $G_0$. Each pair gives the observation
\[
S=\frac{d_{G_0}(U,V)}{dist(U,V)}.
\]
Weighted shortest paths respect the actual traversal edges. Unreachable pairs have stretch $+\infty$ and remain in the sample; coincident points use the convention in Appendix~\ref{app:implementation}. These observations measure how much the graph can lengthen a route between database records. Using their maximum protects against every observed detour, although rare large values can lead to broad recovery budgets.

For prescribed $\varepsilon_s,\delta_s\in(0,1)$, we use a classical order-statistic tolerance bound~\cite{wilks1941tolerance} to cover at least a $1-\varepsilon_s$ fraction of the stretch population with confidence $1-\delta_s$. For $B$ independent and identically distributed (i.i.d.) observations $S_1,\ldots,S_B$, let $S_{(j)}$ denote the $j$-th smallest sampled value. We choose the smallest rank $j$ satisfying
\begin{equation}\label{eq:stretch-rank}
\Pr\{\operatorname{Bin}(B,1-\varepsilon_s)\ge j\}\le\delta_s.
\end{equation}
\begin{proposition}[Sampled stretch coverage]\label{prop:sample-stretch}
Write $r(t)=\Pr(S>t)$ for a fresh observation from the same population. With a prespecified rank satisfying Equation~\eqref{eq:stretch-rank},
$\Pr\{r(S_{(j)})>\varepsilon_s\}\le\delta_s$.
The statement also holds when each observation is replaced by a conservative upper bound.
\end{proposition}
The choice $j=B$ uses the sample maximum and satisfies the coverage condition when
\begin{equation}\label{eq:stretch-budget}
B\ge\left\lceil\frac{\log\delta_s}{\log(1-\varepsilon_s)}\right\rceil.
\end{equation}
For example, 299 observations give 99\% population coverage with 95\% confidence. This controls the fraction exceeding the proposal, not its numerical error relative to the graph's maximum. If no rank is admissible or its selected value is infinite, the requested bound supplies no finite proposal.

\texttt{MAX} uses the sampled maximum directly, with operating value at least one. \texttt{PAIR\_ORDER} uses the admissible order statistic. Pair coverage alone does not control the paths needed by rejected queries; Section~\ref{sec:query-proposals} addresses that distinction.

\begin{table}[t]
\centering
\small
\caption{Stretch estimators. Order-statistic candidates use Equation~\eqref{eq:stretch-rank}; all outputs are fixed before recovery certification.}
\label{tab:estimators}
\begin{tabular}{lp{0.53\columnwidth}}
\toprule
\textbf{Estimator} & \textbf{Observation and selection} \\
\midrule
\texttt{MAX} & Maximum sampled graph-pair stretch. \\
\texttt{PAIR\_ORDER} & Order statistic of graph-pair stretches. \\
\texttt{GEV} & Fitted block return level, bounded below by the sampled maximum. \\
\texttt{QUERY\_PATHS} & Order statistic of exact query-path scores $Z_\tau$. \\
\texttt{RECOVERY\_PATHS} & Order statistic of conservative scores from recorded paths. \\
\bottomrule
\end{tabular}
\end{table}

\subsection{Block-Maxima Estimation}
Order statistics control the fraction of observations above a selected stretch without modeling the tail. As an alternative, the generalized extreme-value (GEV) estimator uses a fitted tail model to propose a stretch from maxima of sampled blocks. We partition $B=mb$ pair stretches into $m$ blocks of size $b$ and take $M_i=\max_{1\le j\le b}S_{i,j}$ in each block. If $F$ is the cumulative distribution function (CDF) of pair stretch, then $\Pr(M_i\le s)=F(s)^b$. We fit a GEV CDF $\widehat H$ to the block maxima by maximum likelihood. For a prespecified $\beta\in(0,1)$, its fitted quantile $\widehat t_\beta=\widehat H^{-1}(\beta)$ is the \emph{block return level}. We store the operating estimate
\begin{equation}\label{eq:offline-stretch}
t_0=\max\{1,\widehat t_\beta,\max_{i,j}S_{i,j}\}.
\end{equation}
The sample-maximum term keeps the stored value at least as large as every observed stretch. Extreme-value theory (EVT) motivates the fitted model~\cite{ft28,col01}, but a fitted block quantile does not certify the maximum over all graph pairs. We use it to allocate search effort and test the resulting recovery outputs independently, just as for the other proposals. A nonfinite sample or invalid fit yields no finite GEV proposal. Appendix~\ref{app:estimators} gives the distribution, return-level formula, and fitting conditions.

\subsection{Recall-Aligned Query-Path Estimation}\label{sec:query-proposals}
Pair-based estimates characterize paths between database records, whereas recovery needs paths from a query to enough of its true neighbors to meet the recall target $\tau$. We therefore estimate stretch using independent proposal queries rejected by the acceptance rule. Both \texttt{QUERY\_PATHS} and \texttt{RECOVERY\_PATHS} use this query population, exact top-$k$ labels, and the same score construction. They differ in how they obtain the path lengths used in that score. The initial search, acceptance rule, and attachment rule are fixed before collecting the proposal scores.

\paragraph{Exact query paths.}
For each proposal query, \texttt{QUERY\_PATHS} computes shortest-path lengths from the virtual source $q$ in $G'_q$ to each of its $k$ true neighbors. It orders these lengths and takes the $h$-th smallest, where $h=\lceil\tau k\rceil$. Dividing this length by the exact $k$-th-neighbor distance $d_k^*(q)$ and taking the maximum with one gives $Z_\tau(q)$ in Equation~\eqref{eq:query-stretch-score}. The score thus measures a stretch sufficient to reach at least $h$ true neighbors; the neighbors are ordered by path length for this calculation.

\paragraph{Recorded recovery paths.}
\texttt{RECOVERY\_PATHS} obtains path lengths from a fixed proposal recovery traversal in the same augmented graph $G'_q$. For each true neighbor, it uses the length of a path recorded by that traversal, or $+\infty$ if none was found. It then applies the same ordering and normalization as \texttt{QUERY\_PATHS}, including the conventions for $d_k^*(q)=0$. Each recorded length upper-bounds the corresponding shortest-path length, so the resulting score upper-bounds $Z_\tau(q)$. This method avoids separate shortest-path computations, but may produce a larger or infinite score. The proposal traversal is fixed before its scores are observed.

For either estimator, we obtain one score per rejected proposal query and select the operating stretch using the order-statistic bound in Equation~\eqref{eq:stretch-rank}. Queries are the independent observations; paths to different neighbors of one query are not treated as independent samples. Under the conditions of Proposition~\ref{prop:partial-recall}, a finite selected stretch has recovery-failure probability at most $\varepsilon_s$ on the rejected-query population, with confidence $1-\delta_s$ over proposal sampling. The conservative scores of \texttt{RECOVERY\_PATHS} preserve this guarantee. With additional work limits, stretch-score coverage still holds, but the corresponding recall guarantee no longer follows; independent recovery certification assesses the limited procedure. If no admissible rank exists or the selected score is infinite, there is no finite proposal. Section~\ref{sec:offline-recovery} gives the construction, and Section~\ref{sec:recovery-calibration} independently tests the resulting recovery procedure.

\subsection{Budget Construction and Offline Cost}\label{sec:offline-recovery}
Algorithm~\ref{alg:stretch-proposal} turns these observations into the stored recovery budget. We specify the sampling plan and rank before inspecting scores, and use only proposal queries rejected by the fixed acceptance rule. Their count can therefore be smaller than the allocated proposal pool. The selected stretch is shared across deployment queries, while each query's attachments and evolving cutoff determine its own search region.

\begin{small}
\begin{algorithm}[t]
\caption{Offline query-level stretch estimation}\label{alg:stretch-proposal}
\begin{algorithmic}[1]
\REQUIRE Fixed acceptance and attachment rules, $B$ independent rejected proposal queries, exact labels, target $\tau$, budgets $\varepsilon_s,\delta_s$, path method
\STATE Choose the smallest $j\le B$ satisfying Equation~\eqref{eq:stretch-rank}
\IF{no admissible rank exists}
  \RETURN No finite coverage-controlled proposal
\ENDIF
\FOR{each proposal query $q_i$}
  \STATE Obtain exact shortest-path lengths, or upper bounds recorded by the fixed proposal traversal
  \STATE Compute $z_i$ using Equation~\eqref{eq:query-stretch-score}
\ENDFOR
\STATE $t'\gets z_{(j)}$
\RETURN $t'$ if finite; otherwise no finite proposal
\end{algorithmic}
\end{algorithm}
\end{small}

With metric edge weights available, $B$ single-source Dijkstra traversals using a binary heap cost $O(B(|E|+n)\log n)$ in the worst case; repeated sources can share a traversal. Exact labels and metric distances add preparation cost. Block formation costs $O(B)$, and each GEV likelihood evaluation costs $O(m)$. Recording proposal-recovery paths avoids separate shortest-path traversals for those queries. These costs occur offline; deployment reads the stored value.

We use query-path estimation for the primary CTR policy and compare all five estimators in Section~\ref{sec:exp-stretch}. Each estimate and its complete recovery procedure are fixed before collecting recovery-calibration data; a procedure may use exact scan if estimation supplies no finite proposal. The resulting policy can execute additional search using only information available at deployment. Its population reliability still depends on the quality of both the accepted initial results and the recovered outputs. We address these two sources of failure through independent output tests.

\section{Population Certification and Guarantee}\label{sec:conformal_certification}
We now calibrate the acceptance rule and recovery policy introduced in Section~\ref{sec:architecture}, with the goal of bounding the probability that CTR returns recall below $\tau$. Recall that \rev{the index search} produces an initial result $H(q)$ and a search trace $\mathcal T(q)$ for query $q$. The acceptance rule $C(q)$ uses this information to decide whether CTR returns $H(q)$ immediately, when $C(q)=1$, or invokes the stored recovery policy, when $C(q)=0$. We use the binary loss $L_\tau$ to measure failures among accepted and rejected queries separately, then combine the two bounds into a guarantee for the complete policy.

Calibration must select these components and provide a risk bound that holds with prescribed confidence over the calibration data. The guarantee must remain valid after selection. In its standard form, Conformal Risk Control (CRC)~\cite{angelopoulos2024conformal} bounds expected loss averaged jointly over the calibration sample and a fresh observation, using a loss that is monotone in the tuning parameter. This expectation guarantee does not by itself provide the high-confidence risk bound for the selected policy that we require. Moreover, the failure probability conditional on acceptance need not be monotone in the score threshold.

We therefore use Learn-then-Test (LTT)~\cite{angelopoulos2022learn} as CTR's calibration rule for both initial-result acceptance and recovery. At each stage, candidate acceptance rules or recovery procedures are specified using data independent of the corresponding calibration sample. LTT tests each candidate against the null that its risk is at least the allowed level. A multiple-testing correction controls the probability of certifying any candidate whose risk exceeds that level, so selecting among passing candidates preserves the guarantee without requiring monotone risks. For our binary loss, the individual tests use exact binomial tail probabilities. Section~\ref{sec:ltt-acceptance} applies this rule to select $C$ and control failures of $H(q)$ among accepted queries. Section~\ref{sec:recovery-calibration} then holds $C$ fixed and applies the same calibration rule to recovery results on rejected queries. Section~\ref{sec:end-to-end} combines the resulting conditional guarantees.

\subsection{Initial-Result Acceptance}\label{sec:ltt-acceptance}
We predict whether $H(q)$ meets the recall target $\tau$ using the logistic score $s(q)=\sigma(b+w^\top z(q))$, where $\sigma(a)=(1+e^{-a})^{-1}$ and $z(q)$ is a feature vector computed from $H(q)$ and $\mathcal T(q)$. The features include the first 100 sorted returned distances, the ten smallest remaining frontier distances, the number of evaluated records, and the count and cumulative magnitude of reversals in successive expansion distances. Missing distance features are zero-padded consistently (Appendix~\ref{app:score}). We fit $b$ and $w$ on score-training queries. A more informative score can allow CTR to accept more queries; the independent calibration test supplies the risk bound.

With the fitted score fixed, we use the development split to specify a finite set of thresholds $\Theta$. Each threshold defines a candidate acceptance rule
\begin{equation}\label{eq:acceptance-rule}
C_\theta(q)=\mathbf1\{s(q)>\theta\}.
\end{equation}
We evaluate these candidate rules on $n_a$ new i.i.d.\ workload queries with exact labels. For each $\theta$, let $m_\theta=\sum_{i=1}^{n_a}C_\theta(q_i)$ be the number of accepted queries and $x_\theta=\sum_{i=1}^{n_a}C_\theta(q_i)L_\tau(H(q_i),q_i)$ be the number whose initial results fall below $\tau$. The test therefore concerns failures among accepted queries, rather than among all calibration queries.

For a tolerated failure fraction $\epsilon\in(0,1)$ and calibration error probability $\alpha\in(0,1)$, we test the null that the population failure probability conditional on acceptance is at least $\epsilon$. For each candidate, we use $p(x_\theta,m_\theta;\epsilon)$, where
\begin{equation}\label{eq:binomial-test}
p(x,m;\epsilon)=\sum_{j=0}^{x}\binom mj\epsilon^j(1-\epsilon)^{m-j},
\end{equation}
with $p=1$ for $m=0$. This lower-tail binomial p-value is valid for the stated null.

LTT~\cite{angelopoulos2022learn} tests each candidate at level $\alpha/|\Theta|$, accounting for selection across thresholds. We select the smallest passing threshold $\hat\theta$, which maximizes acceptance among the certified candidates, and store $C=C_{\hat\theta}$ for query execution. If none passes, or $\Theta$ is empty, we set $\hat\theta=+\infty$; the stored rule then rejects every query and sends it to recovery.
\begin{theorem}[Initial acceptance]\label{thm:ltt}
Conditional on the score, thresholds, index, search settings, and risk parameters chosen before acceptance calibration, suppose calibration queries and deployment queries follow the same i.i.d.\ law and calibration labels are exact. With probability at least $1-\alpha$ over calibration, the selected rule either accepts no population queries or satisfies
\[
\Pr\{L_\tau(H(Q),Q)=1\mid C(Q)=1\}\le\epsilon.
\]
\end{theorem}
Thus the score determines how effectively CTR separates queries, while calibration determines which threshold has sufficient evidence for deployment. The finite-family correction accounts for this selection without requiring independent thresholds. Appendix~\ref{app:proof-ltt} gives the proof.

\subsection{Recovery Certification}\label{sec:recovery-calibration}
After acceptance calibration fixes $C$, we must decide whether graph recovery is reliable enough to use for rejected queries. We specify a tolerated failure fraction $\epsilon_r\in(0,1)$ and a calibration error probability $\delta_r\in(0,1)$. Our goal is to certify, with confidence $1-\delta_r$ over calibration, that at most an $\epsilon_r$ fraction of future rejected queries return recall below $\tau$ after recovery. If graph recovery cannot be certified, CTR uses exact scan for these queries.

We prepare a finite family $\Lambda$ of $J\ge1$ candidate recovery configurations. Each configuration $\lambda$ specifies an operating stretch $t'$ (Section~\ref{sec:sampling-stretch}), the pool of evaluated candidates to which the query may connect, and the maximum number $M_0$ of these connections (Section~\ref{sbeq}). It also specifies any deterministic search-work limits added to the SBE-Q traversal with MBV in Algorithm~\ref{alg:sbe_mbv_corrected}, as discussed in Section~\ref{sec:recovery-properties}. Given a rejected query $q$, configuration $\lambda$ starts from the saved $H(q)$ and $\mathcal T(q)$ and returns a neighbor set denoted by $\mathcal R_\lambda(q)$.

These configurations are fixed before recovery calibration. Let $\mathcal B$ collect the acceptance rule $C$, the candidate family $\Lambda$, their settings and stretch estimates, and the data used to choose them. A guarantee conditional on $\mathcal B$ concerns these already chosen configurations. The calibration queries provide new evidence about their reliability and are not used to change their behavior.

We draw a prespecified number $N_r$ of new i.i.d.\ queries from the deployment workload, independently of $\mathcal B$. We run \rev{the index search} and apply $C$ to each query, retaining those for which $C(q_i)=0$. Their index set is $I_r=\{i:C(q_i)=0\}$ and their count is $n_r=|I_r|$. Thus $N_r$ counts all queries drawn, while $n_r$ counts the rejected queries on which recovery is tested. We obtain exact top-$k$ labels for each retained query using an independent oracle, and run every candidate from the same saved $H(q_i)$ and $\mathcal T(q_i)$. Any randomness in the initial \rev{index} search is included in this query observation and retained trace.

For recovery configuration $\lambda$, let $e_\lambda$ be the number of rejected calibration queries for which it fails to provide an answer with recall at least $\tau$. A returned set $\mathcal R_\lambda(q_i)$ with recall below $\tau$ contributes one failure through the binary loss $L_\tau(\mathcal R_\lambda(q_i),q_i)$. Invalid or undelivered answers also count as failures unless an exact fallback completes. Let $r_\lambda^{\mathrm{rec}}(\mathcal B)$ denote the probability that configuration $\lambda$ fails to meet this recall requirement on a fresh query rejected by $C$:
\begin{equation}\label{eq:recovery-risk}
\begin{aligned}
e_\lambda&=\sum_{i\in I_r}L_\tau(\mathcal R_\lambda(q_i),q_i),\\
r_\lambda^{\mathrm{rec}}(\mathcal B)
&=\Pr_Q\{L_\tau(\mathcal R_\lambda(Q),Q)=1\mid C(Q)=0,\mathcal B\}.
\end{aligned}
\end{equation}

The observed failure fraction alone does not establish the population bound: even zero failures can provide insufficient evidence when $n_r$ is small. As in Section~\ref{sec:ltt-acceptance}, LTT tests the null $r_\lambda^{\mathrm{rec}}(\mathcal B)\ge\epsilon_r$. We denote the recovery-test p-value for configuration $\lambda$ by $p_\lambda^{\mathrm{rec}}$. To compute it, we use the binomial-tail function $p(x,m;\epsilon)$ from Equation~\eqref{eq:binomial-test}, substituting $e_\lambda$ observed recovery failures for $x$, $n_r$ rejected calibration queries for $m$, and the tolerated failure fraction $\epsilon_r$ for $\epsilon$. The resulting value is the probability of observing at most $e_\lambda$ failures in $n_r$ independent trials if the population failure probability were $\epsilon_r$. A small value means that so few failures would be unlikely at this risk level, providing evidence against the null and supporting certification.

Because we may choose among $J$ candidates, we must control the chance of certifying any candidate whose risk exceeds $\epsilon_r$. We therefore test each at level $\delta_r/J$. The p-values and set of passing configurations are
\begin{equation}\label{eq:recovery-test}
p_\lambda^{\mathrm{rec}}=p(e_\lambda,n_r;\epsilon_r),\qquad
\widehat\Lambda=\{\lambda:p_\lambda^{\mathrm{rec}}\le\delta_r/J\}.
\end{equation}
When $n_r=0$, we set the p-values to one, so no candidate passes. The budgets $\epsilon_r,\delta_r$, candidate family, and sample size are specified before observing recovery outcomes.

If $\widehat\Lambda$ is nonempty, we select one passing configuration $\hat\lambda$, for example the fastest among those certified. CTR then uses $\mathcal R_{\hat\lambda}$ for every future rejected query. If no candidate passes, CTR uses exact scan, which returns the true top-$k$ set. This decision is made once during offline preparation; query execution applies the selected recovery policy whenever $C(q)=0$.
\begin{theorem}[Recovery-population certification]\label{thm:recovery-calibration}
Conditional on $\mathcal B$, with probability at least $1-\delta_r$ over the independent recovery-calibration sample, every passing procedure has $r_\lambda^{\mathrm{rec}}(\mathcal B)\le\epsilon_r$. The stored recovery policy, including exact scan when no candidate passes, satisfies the same bound. If rejected queries have zero population probability, recovery contributes zero error.
\end{theorem}
The result follows because, with $C$ and $\mathcal B$ fixed, the retained queries represent the workload that will reach recovery. Conditional on their count $n_r$, they are independent samples from that workload, so each candidate's failure count is binomial. The correction by $J$ gives a simultaneous guarantee for all passing candidates, even though they share the calibration queries. This is why choosing among them preserves the bound. Algorithm~\ref{alg:recovery-calibration} summarizes the steps, and Appendix~\ref{app:proof-recovery} gives the proof. Appendix~\ref{app:calibration} provides further sampling details and extensions to randomized recovery and fallback within a candidate.

\begin{footnotesize}
\begin{algorithm}[t]
\caption{Offline recovery certification for rejected queries}\label{alg:recovery-calibration}
\begin{algorithmic}[1]
\REQUIRE Fixed rule $C$ and preceding information $\mathcal B$, candidate family $\Lambda$, fixed query count $N_r$, target $\tau$, budgets $\epsilon_r,\delta_r$
\STATE Draw $N_r$ new independent workload queries $q_1,\ldots,q_{N_r}$
\STATE $I_r\gets\emptyset$; $\mathcal D_r\gets()$; $J\gets|\Lambda|$
\FOR{$i=1,\ldots,N_r$}
  \STATE Run \rev{the index search} to obtain $H(q_i)$ and $\mathcal T(q_i)$
  \IF{$C(q_i)=0$}
    \STATE Add $i$ to $I_r$; compute $\mathcal N_k^*(q_i)$ with the exact oracle
    \STATE Store $q_i$, $H(q_i)$, $\mathcal T(q_i)$, and its exact labels in $\mathcal D_r$
  \ENDIF
\ENDFOR
\STATE $n_r\gets|I_r|$
\FORALL{$\lambda\in\Lambda$}
  \STATE Run candidate $\mathcal R_\lambda$ from each retained result and trace in $\mathcal D_r$
  \STATE $e_\lambda\gets$ number of recovery failures in $\mathcal D_r$
  \STATE Compute $p_\lambda^{\mathrm{rec}}$ by Equation~\eqref{eq:recovery-test}, or set it to one if $n_r=0$
\ENDFOR
\STATE $\widehat\Lambda\gets\{\lambda\in\Lambda:p_\lambda^{\mathrm{rec}}\le\delta_r/J\}$
\IF{$\widehat\Lambda=\emptyset$}
  \RETURN Exact scan
\ELSE
  \STATE Select any $\hat\lambda\in\widehat\Lambda$
  \RETURN Fixed recovery procedure $\mathcal R_{\hat\lambda}$
\ENDIF
\end{algorithmic}
\end{algorithm}
\end{footnotesize}

\paragraph{The role of stretch estimation.}
Stretch estimation determines how much graph search a candidate performs. Recovery certification checks the quality of the answers that search actually returns. Failures caused by insufficient stretch or work limits enter $e_\lambda$; reaching a limit is not itself a failure if the returned answer meets $\tau$. The same test therefore applies to any stretch estimator fixed before calibration, including GEV and recorded-path estimates. Its guarantee concerns recovery outcomes and does not require the estimated stretch to cover every relevant path. The end-to-end bound in Section~\ref{sec:end-to-end} combines this recovery certificate with the acceptance certificate and also applies to fixed work-limited candidates under the same sampling assumptions. Changing a candidate after calibration, or enlarging the sample in response to observed failures, requires a separate validity argument.

\subsection{End-to-End Population Guarantee}\label{sec:end-to-end}
Every deployment query follows one of the two certified branches. We can therefore weight their conditional failure bounds by the fractions of queries they receive. Let $\mathcal A$ contain the fitted score, selected acceptance rule, and selected recovery configuration after calibration. We write $\rho_{\mathcal A}=\Pr(C(Q)=1\mid\mathcal A)$ for the acceptance probability and $R_{CTR}(q)$ for the final recall. The resulting workload guarantee is:
\begin{theorem}[Population compliance of CTR]\label{thm:end-to-end}
Under the conditions of Theorems~\ref{thm:ltt} and~\ref{thm:recovery-calibration}, CTR satisfies the following bound with calibration confidence at least $\max\{0,1-\alpha-\delta_r\}$:
\begin{equation}\label{eq:population-ltt}
\begin{aligned}
\Pr_Q\{R_{CTR}(Q)\ge\tau\mid\mathcal A\}
&\ge\rho_{\mathcal A}(1-\epsilon)\\
&\quad+(1-\rho_{\mathcal A})(1-\epsilon_r).
\end{aligned}
\end{equation}
Consequently, population compliance is at least $1-\max\{\epsilon,\epsilon_r\}$ on this calibration event. Zero-probability branches contribute zero.
\end{theorem}
\begin{proof}[Proof sketch]
Let $G_a$ denote the event that the selected acceptance rule has accepted-query failure probability at most $\epsilon$, and let $G_r$ denote the event that the selected recovery policy has rejected-query failure probability at most $\epsilon_r$. The first test gives $\Pr(G_a^c)\le\alpha$. The second gives $\Pr(G_r^c\mid\mathcal B)\le\delta_r$ for every fixed pre-calibration choice, hence $\Pr(G_r^c)\le\delta_r$ after averaging over those choices. A union bound therefore yields $\Pr(G_a\cap G_r)\ge1-\alpha-\delta_r$, without requiring independence of the events.

Once the selected policies are fixed on $G_a\cap G_r$, a fresh query follows the acceptance branch with probability $\rho_{\mathcal A}$. This branch contributes at most $\rho_{\mathcal A}\epsilon$ to the overall failure probability, while recovery contributes at most $(1-\rho_{\mathcal A})\epsilon_r$. Subtracting their sum from one gives Equation~\eqref{eq:population-ltt}. Exact fallback contributes zero recovery failure. The argument controls the quality of returned results directly, without assuming geometric coverage by the estimated stretch. Appendix~\ref{app:proof-end-to-end} gives the full proof.
\end{proof}

For example, $\epsilon=\epsilon_r=0.05$ and $\alpha=\delta_r=0.025$ certify at least 95\% population compliance with at least 95\% joint calibration confidence. The first percentage describes query outcomes under the selected policy; the second describes the reliability of the calibration process that selected it. For one recovery candidate and zero observed failures, certification requires
\begin{equation}\label{eq:recovery-sample-size}
n_r\ge\left\lceil\frac{\log(\delta_r/J)}{\log(1-\epsilon_r)}\right\rceil,
\end{equation}
which is 72 rejected queries for these budgets and $J=1$. Section~\ref{sec:exp-calibration} examines how calibration size affects certification success and deployed throughput.

\begin{corollary}[Population exact identification]\label{cor:exact}
At $\tau=1$, under the conditions of Theorem~\ref{thm:end-to-end} and with at most $k$ distinct returned IDs, Equation~\eqref{eq:population-ltt} also lower-bounds the probability of exact canonical top-$k$ identification, on the same calibration event.
\end{corollary}
The proof appears in Appendix~\ref{app:proof-exact}. Exact identification is therefore the $\tau=1$ case of the same framework. For smaller targets, CTR certifies the requested recall directly and can allocate recovery effort accordingly. In both cases, the statement applies to the calibrated workload with the index, metric, and procedures held fixed.

\section{Experimental Evaluation}\label{sec:experiments}

Our evaluation asks how much work the recovery design saves, how stretch estimation affects that work, and what certification costs in data and throughput. We first compare the recovery mechanisms and the five estimators, then examine the effect of calibration size. We finally evaluate the complete policy against adaptive baselines, including workloads with up to 100M vectors\rev{. Appendix~\ref{app:nsg} also tests CTR on NSG}.

Unless stated otherwise, all experiments run on one node with two AMD EPYC 9655 processors (Zen~5, 192 cores, 384\,MB L3) and 768\,GB of DDR5 memory. Table~\ref{tab:datasets} lists the datasets.
\label{sec:t2i-transform}We include T2I to assess CTR on similarity-based retrieval. Cosine retrieval admits a ranking-equivalent Euclidean metric on unit-normalized vectors. For T2I's native maximum-inner-product task on 200-dimensional embeddings, we use the ranking-preserving asymmetric transformation of \citet{bachrach2014speeding} to convert maximum inner-product search to Euclidean nearest-neighbor search.

In Section~\ref{sec:exp-ablation}, we build HNSW with $M\in\{16,32\}$, $ef_c\in\{100,200\}$, and run search with $ef_s=100$. The ablation crosses all four $(M,ef_c)$ settings with $k\in\{10,50,100\}$; Sections~\ref{sec:exp-stretch}--\ref{sec:exp-scale} fix $M=32$, $ef_c=200$, and $k=100$ on all datasets. Recovery candidates are tested offline with LTT's binomial test (Section~\ref{sec:recovery-calibration}), which bounds the population failure fraction among rejected queries. Deployed-policy timings include exact scan when no candidate passes.

\paragraph{Query allocation.}
All queries for every dataset are drawn from the learning set released with that dataset. Before fitting any component, we allocate 11{,}300 queries per dataset to six disjoint subsets according to Table~\ref{tab:experimental-query-splits}: 1{,}000 for score training, 300 for development, 1{,}000 for acceptance calibration, 1{,}000 for stretch proposal, a 6{,}000-query recovery-calibration pool, and 2{,}000 for final evaluation. We fit the score on the training split and choose candidate thresholds and settings on the development split. Acceptance calibration fixes the rule before we estimate stretches. The independent recovery-calibration pool then tests recovery on queries rejected by that rule. We evaluate the final policy on the held-out queries, using the same set for every method within a dataset. The three timing passes repeat this set and do not create additional independent observations.

The scalability experiments on BIGANN100M, T2I1M, T2I5M, and T2I10M use the same allocation. Only rejected queries contribute query-path observations or recovery-test outcomes; their counts can be smaller than the allocated proposal and calibration pools.

\begin{table}[t]
\centering
\small
\setlength{\tabcolsep}{3pt}
\caption{Query allocation per dataset.}
\label{tab:experimental-query-splits}\label{tab:e2e-query-splits}
\begin{tabular}{lr}
\toprule
\textbf{Role} & \textbf{Queries per dataset} \\
\midrule
Score training & 1{,}000 \\
Development & 300 \\
Acceptance calibration & 1{,}000 \\
Stretch proposal & 1{,}000 \\
Recovery-calibration pool & 6{,}000 \\
Final evaluation & 2{,}000 \\
\midrule
Total & 11{,}300 \\
\bottomrule
\end{tabular}
\end{table}

\begin{table}[t]
\centering
\small
\setlength{\tabcolsep}{3pt}
\caption{Experimental datasets and indexed dimensions. T2I uses the 200-to-201-dimensional lifting. BIGANN100M contains the first 100M SIFT1B vectors; its queries come from the released SIFT learning set.}
\label{tab:datasets}
\begin{tabular}{lrrll}
\toprule
\textbf{Dataset} & \textbf{Base} & \textbf{Dim.} & \textbf{Type} & \textbf{Distance} \\
\midrule
SIFT~\cite{pq} & 1M & 128 & Images & $L_2$ \\
BIGANN~\cite{jegou2011searching} & 100M & 128 & Images & $L_2$ \\
DEEP~\cite{babenko2016deep} & 1M & 96 & Images & $L_2$ \\
GIST~\cite{pq} & 1M & 960 & Images & $L_2$ \\
T2I~\cite{simhadri2022bigann} & 1M/5M/10M & 201 & Text to Images & $L_2$ (lifted) \\
\bottomrule
\end{tabular}
\end{table}

\subsection{Recovery Efficiency}\label{sec:exp-ablation}

We begin by asking whether query-centered expansion and metric pruning reduce the cost of finding missing neighbors. The four configurations in Table~\ref{tab:recovery-methods} vary the expansion source, attachments, and enabled pruning. We evaluate three datasets across four $(M,ef_c)$ settings and three values of $k$. Within each setting, all configurations share the numerical operating value $t=t'$ from the query-path order-statistic proposal at $\tau=0.9$ (Proposition~\ref{prop:sample-stretch}, $1-\delta_s=0.95$), fixed before recovery calibration. Section~\ref{sec:exp-stretch} compares alternative estimators. Calibration and evaluation use Table~\ref{tab:experimental-query-splits}. Each candidate's binomial test has significance level 0.005, giving at least 98\% simultaneous recovery-calibration confidence. Its sample count is $n_r$, the number of rejected calibration queries. For this component ablation, we force recovery on every held-out query and measure recovery-only latency over three passes. Figure~\ref{fig:ablation-time} shows two slices of the grid. All four configurations achieve an observed recall of 1 in these timings. 

SBE-NN and SBE-NN+LBP\&EP expand from the single nearest initial candidate. SBE-Q attaches the virtual query to up to $k$ initial HNSW results ($\mathcal H_{found}=H(q)$, $M_0=k$). Here, ``MBV'' denotes SBE-Q with both pruning tests from Section~\ref{sec:mbv} enabled. Comparisons within each source family assess pruning.

\begin{table}[t]
\centering
\small
\caption{Recovery configurations. LBP: lower-bound pruning; EP: path-budget (elliptical) pruning.}
\label{tab:recovery-methods}
\begin{tabular}{lccc}
\toprule
\textbf{Method} & \textbf{Budget center} & \textbf{LBP} & \textbf{EP} \\
\midrule
SBE-NN & nearest candidate & off & off \\
SBE-NN+LBP\&EP & nearest candidate & on & on \\
SBE-Q & query & off & off \\
MBV & query & on & on \\
\bottomrule
\end{tabular}
\end{table}

\begin{figure}[t]
\centering
\includegraphics[width=\columnwidth]{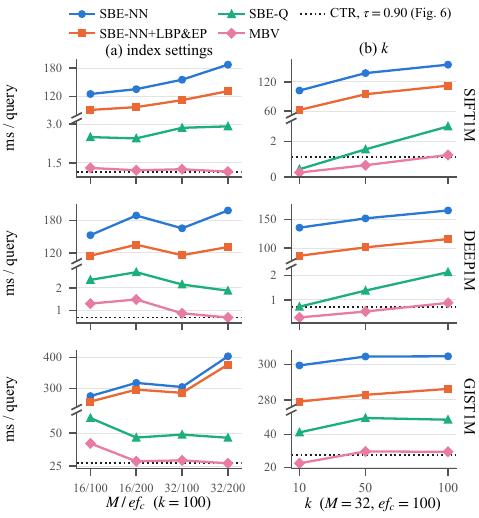}
\caption{Mean recovery-only latency (ms/query). Each panel has a linear $y$-axis broken between the SBE-NN methods (top) and the SBE-Q methods (bottom), with different scales. Rows: datasets. (a) $(M,ef_c)$ at $k=100$. (b) $k$ at $M=32$, $ef_c=100$. Dotted horizontal lines mark CTR's complete-query latency at $\tau=0.90$ in Figure~\ref{fig:e2e}.}
\label{fig:ablation-time}
\Description{Three dataset rows compare four recovery configurations across index settings and requested neighbor counts. Broken axes separate the candidate-source and query-source time ranges.}
\end{figure}

In general across settings, Figure~\ref{fig:ablation-time} shows that query-centered expansion reduces recovery-only latency more than pruning the single-source search. Across the certified configurations with $t>1$, SBE-Q is 4--190$\times$ faster than SBE-NN. This gain reflects the combined change from one candidate source with budget $t(\hat d_1+D)$ to a multiply attached query source with budget $t'D$, where $D$ is the changing result cutoff from Section~\ref{sec:recovery}. Enabling pruning adds at most 1.6$\times$ speedup to SBE-NN and 1.5--2.7$\times$ to SBE-Q. Combining query-centered expansion and both pruning tests (MBV) gives approximately 6.5--493$\times$ speedup over unpruned SBE-NN in the plotted comparisons.

Specifically for various construction parameters, Figure~\ref{fig:ablation-time}(a) fixes $k=100$ and varies $M\in\{16,32\}$ and $ef_c\in\{100,200\}$. SBE-Q's speedup over SBE-NN ranges from approximately 50--64$\times$ on SIFT1M, 65--106$\times$ on DEEP1M, and 4.5--8.7$\times$ on GIST1M. Within each dataset, the speedup increases with either construction parameter when the other is fixed: the blue SBE-NN and orange pruned SBE-NN latencies rise more than the green SBE-Q latency. With larger $M$, the extra edges explored outweigh the benefit of lower graph stretch. With lower $ef_c$, edges are longer on average, so Dijkstra reaches its path-length budget sooner despite a larger estimated stretch. Multiple query attachments reduce dependence on paths from a single candidate, making the green SBE-Q and pink MBV curves less sensitive to these construction changes.

Meanwhile, by ablating $k$ in Figure~\ref{fig:ablation-time}(b), we vary $k\in\{10,50,100\}$ while fixing $M=32$ and $ef_c=100$. Across methods, runtime generally increases with $k$, reflecting the additional search required by a larger recovery radius. MBV remains the most efficient recovery configuration across the tested $k$ values.

\subsection{Evaluation of Stretch Estimation}\label{sec:exp-stretch}

The preceding comparison uses one shared operating stretch to examine recovery design. We now fix recovery to MBV and ask how the estimator changes its cost. We vary only the source of the stored stretch (Table~\ref{tab:estimators}), keeping the index, acceptance rule, and certification test unchanged.

The query-path estimators, \texttt{QUERY\_PATHS} and \texttt{RECOVERY\_PATHS}, use the \emph{stretch-proposal subset} in Table~\ref{tab:experimental-query-splits}: 1{,}000 queries per dataset set aside for estimating the operating stretch. This subset is disjoint from the training, development, acceptance-calibration, recovery-calibration, and evaluation subsets. After acceptance calibration fixes $C$, we run the initial HNSW search on each proposal query and apply this rule. Accepted queries ($C(q)=1$) would return their initial results and do not contribute stretch observations. Rejected queries ($C(q)=0$) would require recovery; each contributes a query-path score computed using its exact neighbors (Section~\ref{sec:query-proposals}). Thus 1{,}000 is the number allocated to this subset, while the number of scores is the number rejected by $C$, which may be smaller. These scores determine the proposed operating stretch; the separate recovery-calibration subset tests the resulting recovery policy.

The pair-based estimators, \texttt{MAX}, \texttt{PAIR\_ORDER}, and \texttt{GEV}, use 1{,}000 sampled graph pairs, which do not consume query-pool allocations.

We compare SIFT1M, DEEP1M, and GIST1M at $\tau\in\{0.8,0.9\}$, with three seeds for the pair-based estimators. We report complete-query latency, including initial search, scoring, and recovery or scan, on the held-out split in Table~\ref{tab:experimental-query-splits} over three timing passes. One acceptance rule per dataset is shared by all estimators and both targets, so all estimators process the same recovery workload within each dataset. The pair-based estimators use stretches independent of $\tau$ and pass recovery certification at both targets, so Figure~\ref{fig:stretch-estimators} shows one bar for each. If a recovery configuration fails certification, rejected queries are scanned. On these held-out queries, every deployed policy achieved a recall of 1 at both targets, so the comparison measures latency at identical observed accuracy.

\begin{figure}[t]
\centering
\includegraphics[width=\columnwidth]{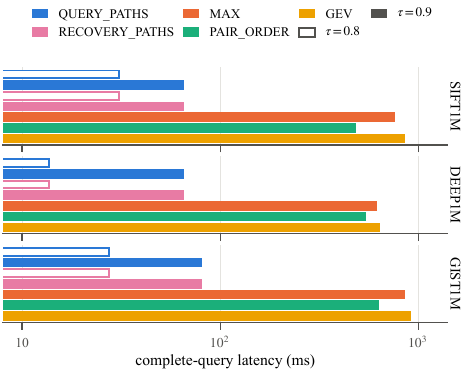}
\caption{Complete-query latency by estimator. Query-path estimators: outlined bars for $\tau=0.8$, filled bars for $\tau=0.9$. Each pair-based estimator has one seed-median bar shared by both targets.}
\label{fig:stretch-estimators}
\Description{Seven horizontal bars per dataset compare two query-path estimators at two recall targets and three pair-based estimators shared by both targets.}
\end{figure}

In Figure~\ref{fig:stretch-estimators}, the short bars of the query-path estimators (blue \texttt{QUERY\_PATHS}, pink \texttt{RECOVERY\_PATHS}) show approximately 7--47$\times$ lower latency than the pair-based bars: query-path estimates focus on nearby true neighbors, yielding lower stretches and tighter budgets while preserving observed accuracy (Table~\ref{tab:proposal-preparation-main}). Both query-path estimators use the same query-neighbor pairs; their matching bar lengths indicate similarly tight budgets in certified settings, while recorded paths avoid separate shortest-path computations offline. With acceptance fixed, the outlined bars at $\tau=0.8$ are shorter than the filled bars at $\tau=0.9$ because the lower target requires smaller recovery budgets. \texttt{MAX}, \texttt{PAIR\_ORDER}, and \texttt{GEV} form a cluster of longer bars because they estimate the upper tail of the same sampled graph-pair stretches, yielding similar stretch values and expansion budgets.

\begin{table}[t]
\centering
\small
\caption{Recorded aggregate estimator ranges across the three 1M datasets. Pair preparation is per seed.}
\label{tab:proposal-preparation-main}
\begin{tabular}{lll}
\toprule
\textbf{Quantity} & \textbf{Query-path} & \textbf{Pair-based} \\
\midrule
Operating stretch & 2.79--3.28 & 3.98--4.59 \\
Preparation (s) & 48--887 & 1{,}739--2{,}593 \\
\bottomrule
\end{tabular}
\end{table}

\subsection{Effect of Calibration Size}\label{sec:exp-calibration}

For each dataset and target $\tau\in\{0.8,0.9\}$, we hold the acceptance rule $C$ and one MBV configuration fixed ($J=1$), then compare prespecified raw calibration sizes $N_r$ from the 6{,}000-query pool (Table~\ref{tab:experimental-query-splits}). Only the $n_r$ rejected queries enter the recovery test, with $\epsilon_r=0.05$ and $\delta_r=0.025$ (Section~\ref{sec:recovery-calibration}). Each batch selects MBV for future rejected queries if its test passes and exact scan otherwise.

Each calibration size has 20 prespecified batch records; full-pool records share one membership and are counted once. Batches at other sizes overlap, so passing fractions describe recorded test decisions rather than independent repetitions or per-query compliance. Each selected policy is evaluated on the same 2{,}000 held-out queries over three passes. QPS measures throughput over complete query execution. Figure~\ref{fig:calibration} reports the fraction of recorded batches that pass and averages the selected policies' QPS across batches, reusing the recorded timings.

\begin{figure}[t]
\centering
\includegraphics[width=\columnwidth]{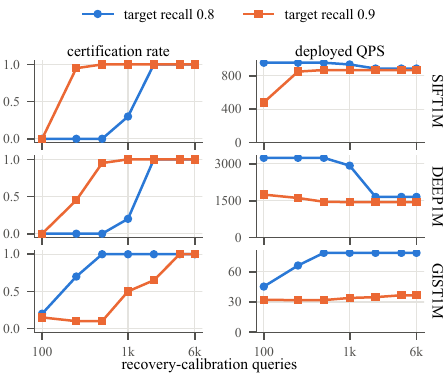}
\caption{Effect of recovery-calibration size. Left: fraction of recorded batches that certify the fixed recovery policy. Right: mean deployed QPS. Each dataset has a 6{,}000-query recovery-calibration pool and 2{,}000 held-out evaluation queries (Table~\ref{tab:experimental-query-splits}).}
\label{fig:calibration}
\Description{Three dataset rows show certification success and deployed throughput as recovery-calibration data increases, for recall targets 0.8 and 0.9. Throughput approaches the fixed graph-recovery policy's throughput as more calibration batches pass.}
\end{figure}

Figure~\ref{fig:calibration} shows how the amount of recovery-calibration data affects certification and deployed throughput. In the left column, larger samples generally increase the fraction of batches that certify the fixed recovery configuration. They provide more evidence about its failure probability; neither the acceptance rule nor the recovery configuration is retrained. On SIFT1M and DEEP1M, certification reaches 100\% sooner at $\tau=0.9$ because stricter acceptance supplies more rejected queries for the recovery test, while recovery failures remain sufficiently rare. GIST1M behaves differently: its acceptance rule rejects every query at both targets, so increasing the target does not increase the number of recovery-test observations. Instead, the higher recovery-failure rate at $\tau=0.9$ requires more observations to establish that the population failure probability is below 5\%.

The right column reflects selection between two fixed deployment policies. A passing batch selects graph recovery for rejected queries; a failing batch selects exact scan. Increasing the calibration sample therefore changes how often each policy is selected, rather than increasing the work performed by either policy on a new query. In the recorded timings, graph recovery is slower than scan on SIFT1M at $\tau=0.8$ and DEEP1M at both targets. As more batches certify graph recovery, their average deployed QPS consequently decreases. Where graph recovery is faster, certification instead increases throughput. Once every batch selects graph recovery, QPS plateaus at that policy's recorded throughput. Certification establishes the required recall reliability; it does not establish that graph recovery is faster than exact scan.

\subsection{End-to-End Performance}\label{sec:exp-e2e}

\begin{figure}[t]
\centering
\includegraphics[width=\columnwidth]{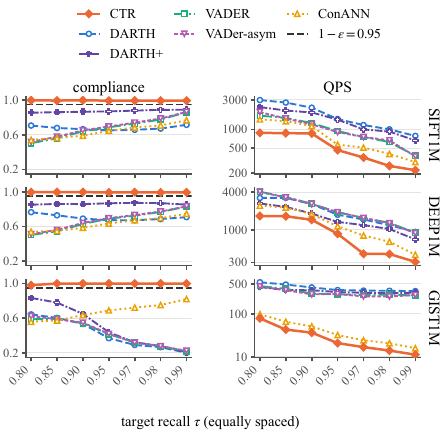}
\caption{End-to-end compliance (left) and throughput (right, log scale) for six methods. Rows: datasets. All seven tested targets are equally spaced and labeled. Compliance is the observed fraction of evaluation queries with recall $\ge\tau$; dashed horizontal lines mark CTR's 0.95 compliance objective. Evaluation uses 2{,}000 held-out queries per dataset (Table~\ref{tab:experimental-query-splits}). $k=100$, three timing passes.}
\label{fig:e2e}
\Description{Three dataset rows compare CTR, DARTH, DARTH+, VADER, VADer-asym, and ConANN in compliance and throughput. All seven recall targets are equally spaced and labeled. A horizontal line marks the compliance objective.}
\end{figure}

We compare CTR against DARTH~\cite{darth}, DARTH+~\cite{darthplus}, VADER~\cite{vader}, VADer-asym, and ConANN~\cite{conann}. The HNSW methods share a base graph built with $M=32$ and $ef_c=200$. We use SIFT1M, DEEP1M, and GIST1M with $k=100$ and seven targets $\tau\in\{0.80,0.85,0.90,0.95,0.97,0.98,0.99\}$. For this comparison, MBV's operating stretch is fixed at $t'=3.5$ across targets. CTR uses LTT acceptance and scans rejected queries when recovery is uncertified. DARTH, DARTH+, VADER, and VADer-asym use predicted recall to stop HNSW search. DARTH uses a 100-tree LightGBM recall predictor with initial and minimum prediction intervals of 1{,}000 and 100, respectively. ConANN uses Faiss IVF-Flat with 1{,}024 inverted lists and target mean false-negative rate $1-\tau$. It probes lists in centroid-distance order and stops using a threshold calibrated on held-out queries. 
We reimplement DARTH+ and VADER because public code was unavailable. DARTH+ uses quantile level $q=0.10$. While VADER uses mean squared recall prediction error, VADer-asym retains VADER's features, architecture, training budget, and stopping rule, replacing only the loss with a heavier penalty for recall overprediction; it is evaluated on every dataset in Figures~\ref{fig:e2e} and~\ref{fig:scale}. Observed compliance is the fraction of evaluation queries with recall $\ge\tau$; CTR targets population compliance $\ge0.95$.

\textbf{The compliance curves compare methods with different statistical guarantees.} 
CTR uses independent LTT calibration to certify, with prescribed confidence, that a specified fraction of future queries meets the recall target, assuming independent queries from the same workload (Section~\ref{sec:end-to-end}). ConANN uses conformal risk control to bound expected recall loss under exchangeability, rather than directly certifying the fraction of queries meeting the target~\cite{conann}. DARTH+'s quantile-$0.10$ variant targets 90\% recall-target attainment. It provides quantile-based stopping with empirical coverage evidence but lacks an established distribution-free, finite-sample guarantee at its adaptive stopping time~\cite{darthplus}. DARTH, VADER, and VADer-asym use prediction-based stopping without a finite-sample recall certificate~\cite{darth,vader}; asymmetric training encourages conservative predictions but does not itself establish a guarantee. The experiments test how consistently each method meets the requested recall target in practice. 

Figure~\ref{fig:e2e}'s left column shows CTR's diamonds at compliance $\ge0.98$, versus 0.20--0.89 for baselines, so only CTR exceeds its 0.95 compliance objective in these measurements. VADER and VADer-asym show comparatively stable or rising compliance on SIFT1M and DEEP1M, but none of DARTH, DARTH+, VADER, or VADer-asym exceeds 0.9 compliance; on GIST1M, their four curves drop to about 0.2 at $\tau=0.99$. These measurements fall below CTR's 0.95 compliance objective; DARTH+ also falls below its nominal 0.90 attainment target. ConANN's rising triangles reflect more IVF probes, but its observed compliance remains below 0.95. CTR combines a calibrated certifier that flags unreliable queries with near-exact recovery to meet the compliance objective.

In Figure~\ref{fig:e2e}'s right column, increasing the recall target from $\tau=0.80$ to $0.99$ reduces throughput for every method across all three datasets. CTR's QPS decreases by factors of 4.0, 5.4, and 6.9 on SIFT1M, DEEP1M, and GIST1M, respectively, while maintaining compliance $\ge0.98$. Across the same datasets, the corresponding QPS reduction factors are 1.6--3.8 for DARTH, 1.4--3.8 for DARTH+, 1.7--4.6 for VADER, 1.6--5.2 for VADer-asym, and 4.9--6.1 for ConANN, quantifying the throughput cost of stricter recall targets for each method.

\subsection{Scalability}\label{sec:exp-scale}

\begin{figure}[!tbp]
\centering
\includegraphics[width=\columnwidth]{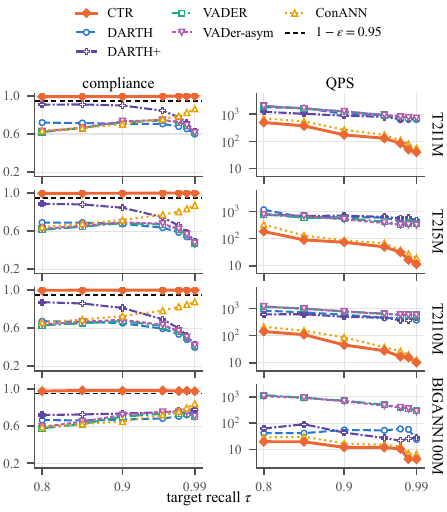}
\caption{Scalability on T2I1M, T2I5M, T2I10M, and BIGANN100M (top to bottom). Left: observed compliance (as in Figure~\ref{fig:e2e}); dashed lines mark CTR's 0.95 compliance objective. Right: QPS (log scale). Targets are on a linear axis, unlike the equally spaced targets of Figure~\ref{fig:e2e}. All rows compare CTR, DARTH, DARTH+, VADER, VADer-asym, and ConANN. $k=100$, three timing passes.}
\label{fig:scale}
\Description{Four dataset rows compare six methods in observed recall compliance on the left and complete-query throughput on the right. Rows show T2I1M, T2I5M, T2I10M, and BIGANN100M, with dataset names along the right side.}
\end{figure}

We test whether Section~\ref{sec:exp-e2e}'s compliance--cost trade-off persists up to 100M vectors. T2I1M, T2I5M, and T2I10M additionally test text queries against image embeddings, a cross-modal setting in which query and database distributions can differ. We use $k=100$, $M=32$, the shared query allocation and held-out queries (Table~\ref{tab:experimental-query-splits}), and the ranking-preserving Euclidean lifting described at the start of this section for every method, with exact labels for the original inner-product task on each subset. All methods use pooled execution time.

Across all tested scales up to 100M, CTR alone has observed compliance above 0.95 at every target. In Figure~\ref{fig:scale}'s top three left panels, CTR maintains observed compliance of at least 0.992 at every size and target, above the 0.95 objective. Across the T2I datasets and targets, VADer-asym improves compliance over VADER by at most 1.62 percentage points, but remains below the 0.95 objective.
At $\tau=0.99$, compliance for DARTH, DARTH+, VADER, and VADer-asym falls from 0.60--0.62 at 1M to 0.40--0.42 at 10M. Their downturns show that predictor-based stopping alone is insufficient at high recall targets on this demanding workload: finding enough true neighbors still requires extensive search (as in Section~\ref{sec:exp-e2e}). ConANN's compliance rises as it searches more clusters, but at $\tau=0.99$ only 87--88\% of queries achieve 99\% recall, below the compliance objective. Although ConANN achieves higher compliance than the other baselines at high $\tau$, without high-confidence recovery it still falls short of the objective.

\paragraph{BIGANN100M} BIGANN100M is the 100M-vector SIFT subset described in Table~\ref{tab:datasets}. We use $k=100$ and the same query allocation and held-out protocol. Figure~\ref{fig:scale}'s bottom-left panel extends the reliability result to 100M: CTR certifies recovery at every target and its diamonds remain above the dashed 0.95 line, with observed compliance of at least 0.9754; no baseline curve exceeds 0.8408.

\paragraph{Throughput.} The right column of Figure~\ref{fig:scale} extends Section~\ref{sec:exp-e2e}'s trade-off to 100M. On T2I, the average 1M-to-10M QPS reduction factor is around 4.0$\times$ for CTR and ConANN, and 1.7$\times$ across the other baselines. The larger reductions for CTR and ConANN are consistent with more conservative calibrated search at larger dataset sizes.

\section{Related Work}\label{sec:related-work}
\paragraph{Graph search and adaptive effort.}
CTR builds on graph-based ANN search, whose indexes include NSW/HNSW~\cite{nsw,malkov2018efficient}, HNSW variants~\cite{acorn,pred2,peng2025dynamic,pred4}, NSG~\cite{nsg}, and DiskANN~\cite{diskann}. Quantization, inverted indexes, and locality-sensitive hashing offer other ways to reduce retrieval cost~\cite{pq,douze2024faiss,johnson2019billion,eucli,angularlsh,mulprobe}; see~\cite{guoliangli} for a survey. Across these approaches, variation in query difficulty motivates hardness measures~\cite{wang2024steiner}, controlled workloads~\cite{ceccarello2025workloads}, and methods that adapt search effort to individual queries~\cite{learnedtermination1,learnedtermination2,learnedtermination3,darth,nongraphivf,chen2021spann,conann,zhang2023auncel,mohoney2025quake}. CTR addresses the accompanying reliability question by combining adaptive acceptance with a specified recovery action and certificates for both branches.

\paragraph{Recovery and statistical calibration.}
Recovering missing neighbors requires deciding where to search and how much additional work is justified. App2Exa uses a separate tree index for recovery~\cite{li2025app2exa}, whereas CTR reuses the \rev{graph index} and its initial search state. MBV also relates to methods that reduce distance calls through geometric bounds~\cite{tribase,trim,acceleratedpruning,jeesavoid,crouting,dco1,betterdco,dade}. To connect this search to a recall requirement, CTR combines classical tolerance bounds~\cite{wilks1941tolerance} for budget estimation with LTT~\cite{angelopoulos2022learn} for output certification. These tools belong to the broader literature on conformal and risk-controlling calibration~\cite{10.1145/3478535,angelopoulos2021gentle,angelopoulos2024conformal}; their roles in CTR are to guide recovery effort and certify the selected policy, respectively.

\section{Conclusion}\label{sec:future-work}
CTR provides a way to use \rev{HNSW and other graph indexes} when an application specifies population recall compliance. It returns initial results when calibration supports doing so and uses the existing graph to recover missing neighbors otherwise. Stretch connects the available distance information to a graph-search budget; recall-aligned estimation and metric pruning make that additional search more focused. Independent output tests establish the reliability of the selected policy, with exact scan when recovery cannot be certified. The resulting framework makes both the recall requirement and the search cost of meeting it explicit.

\FloatBarrier
\clearpage
\nobalance
\bibliographystyle{ACM-Reference-Format}
\bibliography{related}
\clearpage
\appendix
\section{Geometric Analysis}\label{app:geometric}
This appendix collects exact-containment statements that complement the partial-recall result in the main paper. The operating stretches are supplied parameters; independent outcome certification determines the reliability of a deployed procedure. Additional work limits do not change graph-ball containment, but can prevent traversal from reaching the contained neighbors; Lemma~\ref{lem:pruning} therefore excludes these limits.
\subsection{Stretch-Bounded Expansion from a Nearest Candidate (SBE-NN)}
\label{sec:sbe-nn}

Use the graph, cutoff, and tie conventions of Section~\ref{sec:preliminaries}. Write $x_{(1)}=a$ for the nearest initial candidate and $N_k^*(q)=\mathcal N_k^*(q)$. The cutoff bound used in the proof is
\begin{equation}\label{eq:bound_dk}
d_k^*\le\hat d_k.
\end{equation}
Assume a finite $t\ge1$ bounds the stretch from $x_{(1)}$ to each true neighbor; a global spanner bound is sufficient.

\begin{theorem}[Candidate-source containment]\label{thm:sbe-nn}
Assume the stretch bound stated above holds and at least $k$ distinct candidates exist. Given a query $q$, let $x_{(1)}$ be the nearest candidate and $\hat{d}_k$ be the distance to the $k$-th candidate returned by \rev{the initial search}. Every point in the true top-$k$ set $N_k^*(q)$ lies inside the graph ball of radius $t(\hat{d}_1 + \hat{d}_k)$ centered at $x_{(1)}$.

Formally, for every $x^* \in N_k^*(q)$:
\[
d_{G_0}(x_{(1)}, x^*) \le t(\hat{d}_1 + \hat{d}_k)
\]
Consequently:
\[
N_k^*(q) \subseteq \{x \in \mathcal{X} : d_{G_0}(x_{(1)}, x) \le t(\hat{d}_1 + \hat{d}_k)\}
\]
\end{theorem}

\noindent The proof appears in Appendix~\ref{app:proof-sbe-nn}.

\subsection{Query-Source Containment}\label{app:query-source}
\begin{theorem}[Stretch-Bounded Expansion from Query]
\label{thm:sbe-q}
Let $G'$ be the finite augmented traversal graph with virtual source $q$. Suppose a finite $1\le t'<+\infty$ satisfies $d_{G'}(q,x^*)\le t' dist(q,x^*)$ for every $x^*\in N_k^*(q)$. Let $\hat{d}_k$ be the distance to the $k$-th distinct evaluated candidate, or $+\infty$ until $k$ candidates are available. Every true top-$k$ point lies inside the graph ball of radius $t'\hat d_k$ centered at $q$. A global spanner bound is sufficient, but only these ordered query-to-target bounds are required; paths back to the virtual source are unnecessary.

Formally, for every $x^* \in N_k^*(q)$:
\begin{equation}
\label{eq:sbe-q-bound-exact}
d_{G'}(q, x^*) \le t' \cdot \hat{d}_k
\end{equation}
\end{theorem}

\noindent The proof appears in Appendix~\ref{app:proof-sbe-q}.

\subsection{Pruning and Retention}
\begin{lemma}[Preservation of paths to true neighbors]\label{lem:pruning}
On a finite, fixed graph with exact nonnegative metric weights, fix a finite $1\le t'<+\infty$ for the whole traversal and suppose $d_{G'}(q,x^*)\le t' dist(q,x^*)$ for every true top-$k$ neighbor $x^*$. Starting with any distinct candidates and their exact query distances, Algorithm~\ref{alg:sbe_mbv_corrected} preserves all true answers using the dynamic cutoff $g(u)>t'D$ and the two MBV rules above. There is no additional limit on beam width, distance evaluations, or vertex expansions. Ties are resolved by the fixed total order, and all lower bounds are valid in exact arithmetic.
\end{lemma}
\noindent The proof appears in Appendix~\ref{app:proof-pruning}.
The argument in Appendix~\ref{app:proof-pruning} also applies to filtered targets when $D$ is computed from matching candidates and traversal is allowed through all vertices.

For the single-source SBE-NN analysis, keep $a=x_{(1)}$ fixed, write $g_a(u)$ for a path length from $a$, and use $R=t(dist(q,a)+D)$. The corresponding elliptical test is $g_a(u)+LB(u)>R+D$; the query-source test $(t'+1)D$ cannot be used unchanged. A shortest witness path to a true target has length at most $t(dist(q,a)+d_k^*)$, so the same prefix argument proves safety. For filtered targets, the same formula uses $D$ obtained only from matching candidates.

\section{Implementation and Calibration Details}\label{app:implementation}
\subsection{Metric and Record Conventions}

The record, cutoff, and metric conventions are defined in Section~\ref{sec:preliminaries}. An implementation that groups duplicate locations must preserve member IDs and justify its resulting traversal graph and stretch bounds; replacing a group by an arbitrary member's adjacency or predicate is not justified. Predicate membership is checked per ID.

Squared-$L_2$ scores must be square-rooted before use in metric bounds or path weights. Floating-point implementations additionally require conservative error bounds for pruning and consistent tie handling; the proofs use exact metric arithmetic.

\paragraph{T2I transformation.}\label{app:t2i-transform}
T2I's native retrieval objective is maximum inner product. The asymmetric Euclidean transformation of \citet{bachrach2014speeding} preserves this ranking while supplying the metric geometry used by recovery. For each database, we choose a positive constant $R_0\ge\max_x\|x\|_2$ and keep it fixed throughout index construction, estimation, calibration, and evaluation. We map database vectors and queries into $\mathbb R^{201}$ as follows:
\begin{equation}
P(x)=\bigl(x,\sqrt{R_0^2-\|x\|_2^2}\bigr),\qquad Q(q)=(q,0).
\label{eq:t2i-lifting}
\end{equation}
Every lifted database vector has norm $R_0$, and
\begin{equation}
\|Q(q)-P(x)\|_2^2=\|q\|_2^2+R_0^2-2\langle q,x\rangle.
\label{eq:t2i-ranking}
\end{equation}
For a fixed query, the first two terms are constant, so minimizing lifted $L_2$ distance gives exactly the original maximum-inner-product top-$k$ ranking under the common ID tie order. No equal-norm assumption on the original database is needed, and we do not individually normalize its vectors. The maps are asymmetric, but distances in the shared lifted space are ordinary symmetric Euclidean distances.

Graph construction, CTR, and every baseline (DARTH, VADER, and ConANN) use this same lifted representation. Database edges have weights $\|P(x)-P(y)\|_2$, and query attachments and candidate distances use $\|Q(q)-P(x)\|_2$. Graph weights, path lengths, stretch ratios, and metric pruning all use unsquared $L_2$. Exact labels maximize the original inner product over the same database subset, equivalently minimizing lifted $L_2$; corresponding ground-truth neighbor IDs can therefore be reused, while distance values must be converted using Equation~\eqref{eq:t2i-ranking}. For each smaller database subset, we use exact labels for that subset rather than assuming the full-database top-$k$ list remains valid.
Exact fallback may rank the original inner products directly; all distances used for graph traversal and pruning remain lifted, unsquared $L_2$ distances.

\subsection{Initial-Result Score}\label{app:score}
 
We compute the score $s(q)$ from \rev{the index's} search trace. Its feature vector $z(q)$ describes the candidate distances, the extent of the search, and changes in traversal distance:
\[
\begin{aligned}
z(q)=\Big[h_1,\ldots,h_{100},\ f_1,\ldots,f_{10},|\mathcal{T}(q)|,\ n_{\mathrm{rev}}(q),\ \Delta_{\mathrm{rev}}(q)\Big].
\end{aligned}
\]
$h_1,\ldots,h_{100}$ contain the first $\min(|H(q)|,100)$ sorted returned distances and zeros in the remaining entries. $f_1,\ldots,f_{10}$ contain the ten smallest remaining candidate-frontier distances at termination, also zero-padded if fewer remain. These finite feature coordinates are distinct from the search cutoff $\hat d_k$, which is $+\infty$ when fewer than $k$ candidates exist. Padding and missing-result conventions must be identical in training, calibration, and deployment. These features capture distance geometry. Following previous work~\cite{ ceccarello2025workloads}, we capture the search reach with $|\mathcal{T}(q)|$, where $\mathcal{T}(q)$ is the search trace defined in Section~\ref{sec:preliminaries}. We use \emph{distance-reversal} statistics for path stability:
$n_{\mathrm{rev}}(q)$ is the number of direction reversals in successive expansion distances,
and $\Delta_{\mathrm{rev}}(q)$ is their cumulative magnitude. 
More precisely, let
\[
d^{\mathrm{pop}}_1,d^{\mathrm{pop}}_2,\ldots,d^{\mathrm{pop}}_T
\]
denote the sequence of metric distances from $q$ of the nodes popped from the \rev{base-graph} candidate queue \rev{(the bottom layer for HNSW)}. A \emph{distance reversal} occurs at index $j\in\{2,\ldots,T-1\}$ whenever the local direction of this sequence changes, i.e.,
\[
\bigl(d^{\mathrm{pop}}_j-d^{\mathrm{pop}}_{j-1}\bigr)\bigl(d^{\mathrm{pop}}_{j+1}-d^{\mathrm{pop}}_j\bigr)<0.
\]
Equivalently, the popped distances are locally non-monotone; for example,
$100,50,120$ and $50,120,80$ each contain a reversal. Let $r_j(q)$ indicate
whether a local reversal occurs at position $j$. We define
$n_{\mathrm{rev}}(q)=\sum_{j=2}^{T-1} r_j(q)$ as the number of reversals, and
$\Delta_{\mathrm{rev}}(q)=\sum_{j=2}^{T-1}r_j(q)(|d^{\mathrm{pop}}_j-d^{\mathrm{pop}}_{j-1}|+|d^{\mathrm{pop}}_{j+1}-d^{\mathrm{pop}}_j|)$. Empty sums are zero; a shared jump can contribute to two reversals. These features quantify instability in the
\rev{search} traversal order: larger values indicate that the search path is less
geometrically consistent. The contribution of each feature group to score separation and acceptance efficiency requires empirical evaluation.
\subsection{Recovery Calibration Protocol}\label{app:calibration}
\paragraph{Randomized recovery.}
The main text writes $\mathcal R_\lambda(q)$ for the result of deterministic recovery. To allow a candidate to make random choices during recovery, write its result as $\mathcal R_\lambda(q,\omega)$, where $\omega$ records those choices. The distribution of $\omega$ is fixed as part of the candidate configuration before calibration. Fresh per-query randomness is independent across observations and follows the same law during calibration and deployment; shared seeds and fitted artifacts are fixed in $\mathcal B$. Different candidates may share randomness within one query.

For this extension, replace the failure count and risk in Equation~\eqref{eq:recovery-risk} by
\begin{equation}\label{eq:randomized-recovery-risk}
\begin{aligned}
e_\lambda&=\sum_{i\in I_r}L_\tau(\mathcal R_\lambda(q_i,\omega_i),q_i),\\
r_\lambda^{\mathrm{rec}}(\mathcal B)
&=\Pr_{Q,\omega}\{L_\tau(\mathcal R_\lambda(Q,\omega),Q)=1\mid C(Q)=0,\mathcal B\}.
\end{aligned}
\end{equation}
Conditional on $\mathcal B$ and $n_r$, the query observations together with their recovery randomness are independent and identically distributed. Consequently, $e_\lambda$ is binomial with parameters $n_r$ and $r_\lambda^{\mathrm{rec}}(\mathcal B)$. The test in Equation~\eqref{eq:recovery-test} and its correction by $J$ therefore apply without change. The confidence $1-\delta_r$ now concerns both the calibration queries and their recovery randomness. For the deployed policy, let $R_{CTR}(q,\omega)$ denote the final recall. The left-hand side of Equation~\eqref{eq:population-ltt} becomes $\Pr_{Q,\omega}\{R_{CTR}(Q,\omega)\ge\tau\mid\mathcal A\}$; the same lower bound and joint calibration confidence hold.

\paragraph{Fallback within a candidate.}

A candidate may also include an exact-scan fallback for individual queries. The conditions that trigger the scan must be fixed before calibration, and the test uses the final neighbor set returned after any such scan. Calibration evaluates this complete behavior, including the fallback. Graph-recovery work limits do not truncate the exact scan: its zero-error guarantee requires completion. This is a choice within a candidate, whereas the default in Section~\ref{sec:architecture} selects exact scan for every rejected query when no candidate passes certification.

\paragraph{Calibration sample size.}
A fixed total query count produces a random rejected count $n_r$. Alternatively, prespecify a rejected count and collect that many successive rejections from independent workload queries. For rejection probability $\pi>0$, the expected number of raw queries is $n_r/\pi$. Stopping is determined by rejection decisions, not observed recovery successes. Counts, candidate families, and risk budgets are specified before inspecting recovery outcomes. A new index, acceptance rule, stretch proposal, or deployed randomization law requires a new calibration argument. Multiple workloads, $k$ values, or advertised targets require accounting for the complete testing family.

For zero observed failures, Equation~\eqref{eq:recovery-sample-size} gives the necessary recovery sample count. With $\epsilon_r=0.005$ and $\delta_r=0.05$, $J=1$ requires 598 rejected queries and $J=20$ requires 1,196. These count independent query observations, not graph pairs or repeated seeds on one query. For initial acceptance, the analogous zero-failure condition is $m_\theta\ge\lceil\log(\alpha/|\Theta|)/\log(1-\epsilon)\rceil$.

\paragraph{Calibration labels.}\label{rem:oracle}
The main guarantees use exact labels from an independent oracle with the same metric and tie order. If a separate argument bounds the probability of any wrong label across both calibration batches by $\eta$, coupling with the ideal exact-label procedure gives joint confidence at least $1-\alpha-\delta_r-\eta$, clipped at zero.

\section{Estimator Specifications}\label{app:estimators}
Separate advertised sampling guarantees require simultaneous confidence accounting across the reported proposals. Their sample counts and ranks must be specified before observing proposal outcomes; Section~\ref{sec:sampling-stretch} defines the five estimators and their sampling populations.

\subsection{Block-Maxima Proposal}\label{sec:stretch_estimation}

We use the pair population and block maxima defined in Section~\ref{sec:sampling-stretch}. A database with fewer than two records is answered directly without estimating pair stretch.

\paragraph{EVT model and its scope.}\label{thm:evt}
The Fisher--Tippett--Gnedenko result concerns a \emph{nondegenerate limit of suitably normalized maxima}, when such a limit exists~\cite{ft28,col01}. It does not assert that finite block maxima from every distribution are GEV distributed. In particular, sampling from a fixed finite graph gives a discrete finite population whose maxima eventually concentrate on its maximum. We use a GEV fit as a modeling approximation, with CDF
\[
H(s)=
\begin{cases}
\exp\{-[1+\xi(s-\mu)/\sigma]^{-1/\xi}\},&\xi\ne0,\\
\exp\{-\exp[-(s-\mu)/\sigma]\},&\xi=0,
\end{cases}
\]
Here $\mu$ is the location parameter, $\sigma>0$ is the scale, and $\xi$ is the shape. The first expression is defined on $1+\xi(s-\mu)/\sigma>0$, with the usual endpoint extension. Its fitted block return level is
\begin{equation}\label{eq:gev-return}
\widehat t_\beta=
\begin{cases}
\widehat\mu+\dfrac{\widehat\sigma}{\widehat\xi}
 [(-\log\beta)^{-\widehat\xi}-1],&\widehat\xi\ne0,\\
\widehat\mu-\widehat\sigma\log(-\log\beta),&\widehat\xi=0,
\end{cases}
\end{equation}
for $\beta\in(0,1)$. Even with known parameters and a correct block model, $H^{-1}(\beta)$ covers a \emph{new block maximum} with probability $\beta$, not the maximum over all graph pairs. Covering $r$ independent blocks would instead require $H^{-1}(\beta^{1/r})$. Neither formula accounts for parameter estimation or justifies extrapolation to every pair of a fixed graph.

\subsubsection{Offline estimation and use}
\label{sec:offline-estimation}

A valid numerical fit requires finite parameters, positive scale, valid support at every observed block maximum, and finite likelihood and return level. Nonfinite sampled stretches or an invalid fit produce an explicit estimation-failure status; the recovery design specifies the resulting action, including exhaustive search. These checks establish numerical validity, not statistical coverage. A reproducible fit records the graph, pair population, $B$, $b$, $m$, sampled pairs, fitted parameters, failures, and variability across independent offline fits.

The total fitting cost also depends on the optimizer's iteration count and convergence behavior. Section~\ref{sec:offline-recovery} gives the shortest-path, block-formation, and likelihood-evaluation costs.

\rev{\section{NSG Instantiation}\label{app:nsg}
We next test whether CTR carries over to a different graph construction. We build NSG graphs~\cite{nsg} on SIFT1M, DEEP1M, and GIST1M with the authors' implementation from exact 200-nearest-neighbor graphs, using out-degree bound $R=64$, candidate pool $L=200$, and $C=500$, and search with beam width 100 from the navigating node. These settings match the HNSW graph of Section~\ref{sec:exp-e2e}: $R$ equals its bottom-layer degree limit $2M$, the candidate pool equals $ef_c$, and the beam width equals $ef_s$, so the graph construction is the main difference. We retrain the score on NSG search traces, recalibrate acceptance, and recertify recovery, using the query allocation of Table~\ref{tab:experimental-query-splits}, $k=100$, and the seven targets of Section~\ref{sec:exp-e2e}. NSG paths are closer to metric distances than HNSW paths, so we do not reuse $t'=3.5$. On SIFT1M and DEEP1M, for each target, MBV uses the smallest stretch on a descending grid from 3.0 to 2.45 in steps of 0.025 that passes the recovery-certification test; candidates are tested in fixed sequence, so the bound holds at $\delta_r$ without a multiplicity correction. On GIST1M, MBV uses $t'=3.0$. When too few rejected calibration queries remain to certify recovery, or no tested stretch passes the test, rejected queries are scanned. DARTH, DARTH+, VADER, and VADer-asym stop search on the same NSG graph with the DARTH search loop, a beam-width cap of 1{,}000, and predictors trained on the 9{,}300 non-evaluation queries; VADER and VADer-asym use DARTH's features with mean-squared and asymmetric losses. ConANN uses the IVF-Flat index of Section~\ref{sec:exp-e2e}. The NSG experiments run on a different machine from Section~\ref{sec:experiments}: one node with two Intel Xeon Granite Rapids processors (192 cores in total) and 768\,GB of memory. QPS is therefore comparable within Figure~\ref{fig:e2e-nsg} but not with Figure~\ref{fig:e2e}.

\begin{figure}[t]
\centering
\includegraphics[width=\columnwidth]{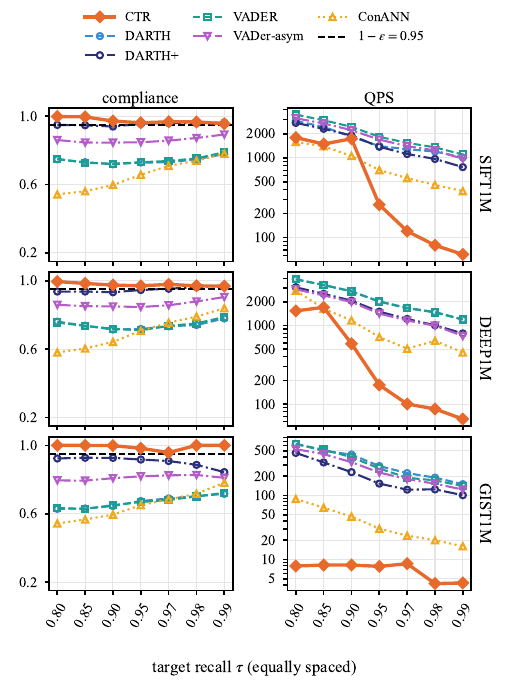}
\caption{End-to-end compliance (left) and throughput (right, log scale) on NSG, as in Figure~\ref{fig:e2e}. Rows: datasets. All seven tested targets are equally spaced and labeled; dashed horizontal lines mark CTR's 0.95 compliance objective. Evaluation uses 2{,}000 held-out queries per dataset, $k=100$, three timing passes.}
\label{fig:e2e-nsg}
\Description{Three dataset rows compare CTR, DARTH, DARTH+, VADER, VADer-asym, and ConANN on NSG in compliance and throughput. A horizontal line marks the compliance objective.}
\end{figure}

CTR extends to the NSG index, and compliance is still upheld. In Figure~\ref{fig:e2e-nsg}'s left column, CTR's diamonds stay at compliance $\ge0.956$ at every target on all three datasets, while DARTH, VADER, VADer-asym, and ConANN remain at 0.54--0.91 and DARTH+ reaches 0.84--0.97 but falls below 0.95 at three or more targets on each dataset. In the right column, CTR's QPS follows the fraction of queries sent to recovery. On SIFT1M and DEEP1M, this fraction rises from almost none at $\tau\le0.85$ to all queries at $\tau=0.99$, and QPS decreases by factors of 29 and 24 from $\tau=0.80$ to $0.99$. On GIST1M, acceptance rejects every evaluation query at every target, and QPS stays between 4.2 and 8.6. The baselines' QPS decreases by factors of 2.9--6.1.}
\FloatBarrier

\section{Proofs}\label{app:proofs}

This appendix gives the proofs of the geometric and population guarantees stated in the main paper.

\subsection{Sampling and partial-recall proofs}\label{app:proof-sampling}
\begin{proof}[Proof of Proposition~\ref{prop:sample-stretch}]
Let $F$ be the stretch CDF, $r(s)=\Pr(S>s)$, and $p=1-\varepsilon_s$. The set $A=\{s:F(s)<p\}$ has probability at most $p$, including for distributions with atoms. If $r(S_{(j)})>\varepsilon_s$, at least $j$ observations lie in $A$. Its probability is therefore at most $\Pr\{\operatorname{Bin}(B,p)\ge j\}$. For $j=B$ this is $p^B$, yielding Equation~\eqref{eq:stretch-budget}. Coordinatewise upper bounds can only increase an order statistic and reduce its exceedance probability. The same argument applies to the distribution of $Z_\tau$ conditional on rejection by the fixed acceptance rule.
\end{proof}
\begin{proof}[Proof of Proposition~\ref{prop:partial-recall}]
At least $h$ true neighbors have witness paths of length $L\le t'd_k^*$. Consider one such neighbor $x$ and a vertex $u$ at prefix length $a$ on its witness path. Since the current result radius satisfies $D\ge d_k^*$, we have $a\le t'D$. The triangle inequality along the remaining directed path also gives $a+dist(q,u)\le L+dist(q,x)\le(t'+1)D$. Thus neither the radius test nor the MBV path-budget test removes this witness. The relaxation argument of Appendix~\ref{app:proof-pruning} ensures that $x$ is reached, evaluated if necessary, and retained under the canonical tie order. Applying it separately to these $h$ neighbors proves recall at least $h/k\ge\tau$. Infinite paths furnish no witness; zero-length paths and duplicate locations follow the stated conventions.
\end{proof}

\subsection{Proof of Theorem~\ref{thm:sbe-nn}}
\label{app:proof-sbe-nn}
\begin{proof}
Let $x^*$ be any point in the true set $N_k^*(q)$. The definition of the true $k$-th distance and Equation~\eqref{eq:bound_dk} give
\begin{equation}\label{eq:one}
\mathit{dist}(q, x^*) \le d_k^* \le \hat{d}_k
\end{equation}

The triangle inequality gives
\begin{equation}\label{eq:two}
\mathit{dist}(x_{(1)}, x^*) \le \mathit{dist}(x_{(1)}, q) + \mathit{dist}(q, x^*)
\end{equation}
Substituting $\mathit{dist}(x_{(1)},q)=\hat d_1$ and the bound in Equation~\eqref{eq:one}, we obtain
\[
\mathit{dist}(x_{(1)}, x^*) \le \hat{d}_1 + \hat{d}_k
\]
This bounds the metric distance between the source $x_{(1)}$ and each true neighbor $x^*$. In Euclidean space equality in the triangle-inequality step is possible when $q$ lies between $x_{(1)}$ and $x^*$; equality in the full bound additionally requires the other inequalities to be tight.

By the assumed stretch bound for the ordered pair $(x_{(1)},x^*)$, the graph distance is bounded by $t$ times the metric distance:
\begin{equation}\label{eq:three}
d_{G_0}(x_{(1)}, x^*) \le t \cdot \mathit{dist}(x_{(1)}, x^*)
\end{equation}
Combining these inequalities yields:
\begin{equation}\label{eq:four}
d_{G_0}(x_{(1)}, x^*) \le t(\hat{d}_1 + \hat{d}_k)
\end{equation}
As $x^*$ was arbitrary, the inclusion holds for all true top-$k$ points.
\end{proof}

\subsection{Proof of Theorem~\ref{thm:sbe-q}}
\label{app:proof-sbe-q}
\begin{proof}
Let $x^*$ be any point in the true neighbor set $N_k^*(q)$. Because $d_k^*\le\hat d_k$, its query distance satisfies
\begin{equation}
\label{eq:metric-bound}
\mathit{dist}(q, x^*) \le d_k^* \le \hat{d}_k
\end{equation}

By the assumed stretch inequality for the ordered pair $(q,x^*)$, the shortest path distance in $G'$ is bounded by the metric distance scaled by $t'$:
\begin{equation}
\label{eq:spanner-bound}
d_{G'}(q, x^*) \le t' \cdot \mathit{dist}(q, x^*)
\end{equation}

Combining Equations~\eqref{eq:metric-bound} and~\eqref{eq:spanner-bound}, we obtain
\begin{equation}
d_{G'}(q, x^*) \le t' \cdot \hat{d}_k
\end{equation}

Since this holds for every $x^*\in N_k^*(q)$, the entire set of true nearest neighbors is contained within the graph ball $B_{G'}(q, t' \cdot \hat{d}_k)$.
\end{proof}

\subsection{Proof of Lemma~\ref{lem:pruning}}
\label{app:proof-pruning}
\begin{proof}
At every step, $D\ge d_k^*$, including the $D=+\infty$ case. Fix a true neighbor $x^*$ and a shortest path to it of length $L\le t'd_k^*$. For a prefix ending at $u$ with length $a$ and remaining suffix length $L-a$, the triangle inequality gives $dist(q,u)\le dist(q,x^*)+L-a$. Thus
\[
a\le t'D,\qquad
a+dist(q,u)\le L+dist(q,x^*)\le(t'+1)D.
\]
The nonstale labels popped by this traversal are nondecreasing: every relaxed label is the current popped label plus a nonnegative edge weight. Omitting relaxations does not change this ordering. Induct along the fixed witness path, starting at $q$. If its next vertex is unsettled, relaxation offers a label no larger than that vertex's prefix length. It will be processed before a radius-based termination, since its queued label is at most $t'd_k^*\le t'D$ even after $D$ decreases. If it was already settled when the predecessor was processed, its popped label was no larger than the predecessor's, hence no larger than its own witness-prefix length. The displayed inequalities show that in either case neither radius pruning nor elliptical pruning can suppress that vertex's outgoing edges when it is processed. The finite graph and strict label improvements ensure termination. At $x^*$, a valid lower bound is at most $dist(q,x^*)\le D$, so its exact query distance is available and it is considered for the result heap. Strict pruning inequalities preserve equal-distance candidates. Every true top-$k$ record is therefore considered, and the final distinct top-$k$ heap is exact.
\end{proof}

\subsection{Proof of Theorem~\ref{thm:ltt}}
\label{app:proof-ltt}
\begin{proof}
For a fixed threshold with positive acceptance probability, let $r_\theta=\Pr\{L_\tau(H(Q),Q)=1\mid C_\theta(Q)=1\}$. Conditional on $m_\theta$, $x_\theta$ is binomial with success probability $r_\theta$. Under the null, its lower tail is no larger than the binomial lower tail at $\epsilon$, so $p_\theta$ is a valid, possibly conservative, one-sided p-value. Bonferroni's inequality bounds the probability of any false rejection by $\alpha$, without independence across thresholds. On the complementary event, every selectable threshold has risk at most $\epsilon$. Reject-all contributes no accepted-query failures.
\end{proof}

\subsection{Proof of Theorem~\ref{thm:recovery-calibration}}
\label{app:proof-recovery}
\begin{proof}
Condition on $\mathcal B$ and $n_r=m$. The acceptance rule is already determined by $\mathcal B$ and is applied unchanged to an i.i.d.\ batch. Conditional on the retained count, the retained observations therefore follow the query law conditional on rejection, independently of one another. Each $\mathcal R_\lambda$ has its behavior, including any work limits and fallback, fixed before this batch. Hence, for each $\lambda$,
\[
e_\lambda\mid(\mathcal B,n_r=m)
\sim\mathrm{Binomial}(m,r_\lambda^{\mathrm{rec}}(\mathcal B)).
\]
Under the null $r_\lambda^{\mathrm{rec}}(\mathcal B)\ge\epsilon_r$, this count stochastically dominates a binomial count with parameter $\epsilon_r$, so $p_\lambda^{\mathrm{rec}}$ is a valid lower-tail p-value. The convention for $m=0$ is conservative. A union bound over the $J$ configurations limits the probability of certifying any configuration with excessive risk to $J(\delta_r/J)=\delta_r$, without requiring independence between configurations. On the complementary event, any selected member is valid. The exhaustive default has zero risk under the exact metric and tie conventions. Averaging over $n_r$ gives the statement conditional on $\mathcal B$. If rejection has probability zero, it occurs neither in calibration nor in deployment, almost surely.
\end{proof}

\subsection{Proof of Theorem~\ref{thm:end-to-end}}
\label{app:proof-end-to-end}
\begin{proof}
Let $G_a$ be the event that the LTT acceptance rule's accepted-query noncompliance is at most $\epsilon$, and $G_r$ the event that the selected recovery has rejected-query noncompliance risk at most $\epsilon_r$. Theorems~\ref{thm:ltt} and~\ref{thm:recovery-calibration} give $\Pr(G_a^c)\le\alpha$ and $\Pr(G_r^c\mid\mathcal B)\le\delta_r$. Thus $\Pr(G_a\cap G_r)\ge1-\alpha-\delta_r$. For each final artifact realization on this event, partitioning by the fixed acceptance rule gives Equation~\eqref{eq:population-ltt} by bounding noncompliance separately in each branch. Independence of the two good events is unnecessary.
\end{proof}

\subsection{Proof of Corollary~\ref{cor:exact}}
\label{app:proof-exact}
\begin{proof}
With the common tie order, recall one and at most $k$ returned IDs imply equality with the true top-$k$ set. Apply Theorem~\ref{thm:end-to-end} at $\tau=1$.
\end{proof}

\section{Additional Experimental Details}\label{app:experimental-details}
\subsection{Operating Stretches and Preparation Costs}\label{app:exp-stretch-details}
The ranges in Table~\ref{tab:proposal-preparation-main} summarize the original three-dataset comparison rather than dataset-specific measurements. Pair-based preparation is measured per seed for the shared geometry, using 1,000 graph pairs and three seeds.

\subsection{Recovery Certification and Timing Details}\label{app:exp-calibration-details}
Table~\ref{tab:experimental-query-splits} reserves 6{,}000 raw queries per dataset for recovery calibration, disjoint from the 1{,}000 score-training, 300 development, 1{,}000 acceptance-calibration, 1{,}000 stretch-proposal, and 2{,}000 evaluation queries. All six subsets are drawn from the learning set released with the corresponding dataset. The calibration-size study draws batches from the recovery-calibration pools; only rejected queries enter the recovery test.

For each dataset, calibration batch sizes vary up to its full recovery-calibration pool. Each size has 20 prespecified batch records; full-pool records share one membership and are counted once. Batches at other sizes overlap. Passing fractions describe recorded test decisions, not independent repetitions or per-query compliance.

QPS reuses each selected policy's measured complete-query timing on the 2{,}000 held-out evaluation queries per dataset, with three timing passes (Table~\ref{tab:experimental-query-splits}). Averaging QPS across calibration batches describes selection between MBV and scan fallback; it is neither pooled throughput nor a new timing replicate.
\FloatBarrier

\end{document}